\documentclass[reqno]{amsart}
\usepackage{amsmath}
\usepackage{amssymb}
\usepackage{paralist}
\usepackage[colorlinks=true]{hyperref}
\hypersetup{urlcolor=blue, citecolor=red}
\newtheorem{theorem}{Theorem}[section]
\newtheorem{proposition}[theorem]{Proposition}
\newtheorem{lemma}[theorem]{Lemma}

\theoremstyle{definition}
\newtheorem{remark}[theorem]{Remark}

\numberwithin{equation}{section}
\numberwithin{theorem}{section}

\newcommand{\mc}[1]{{\mathcal #1}}

\newcommand{\bb}[1]{{\mathbb #1}}

\newcommand{\rme}{\mathrm{e}}

\newcommand{\rmd}{\mathrm{d}}
\newcommand{\eps}{\varepsilon}

\newcommand{\id}{{1 \mskip -5mu {\rm I}}}

\title[A stochastic particle system approximating the BGK equation]{A stochastic particle system approximating the BGK equation}

\author[Paolo Butt\`a and Mario Pulvirenti]{}

\subjclass[2010]
{Primary: 82C40.  
Secondary: 60K35,  
82C22.  
}

\keywords{BGK equation, kinetic limits, stochastic particle dynamics.}

\email{butta@mat.uniroma1.it}
\email{pulviren@mat.uniroma1.it}

\begin{document}
\maketitle

\centerline{\scshape Paolo Butt\`a and Mario Pulvirenti}
\medskip
{\footnotesize
\centerline{Dipartimento di Matematica, Sapienza Universit\`a di Roma} 
\centerline{P.le Aldo Moro 5, 00185 Roma, Italy}
}

\bigskip


\begin{abstract}
We consider a stochastic $N$-particle system on a torus in which each particle moving freely can instantaneously thermalize according to the particle configuration at that instant. Following \cite{BHP}, we show that the propagation of chaos does hold and that the one-particle distribution converges to the solution of the BGK equation. The improvement with respect to \cite{BHP} consists in the fact that here, as suggested by physical considerations, the thermalizing transition is driven only by the restriction of the particle configuration in a small neighborhood of the jumping particle. In other words, the Maxwellian distribution of the outgoing particle is computed via the empirical hydrodynamical fields associated to the fraction of particles sufficiently close to the test particle and not, as in \cite{BHP}, via the whole particle configuration. 
\end{abstract}

\section{Introduction}
\label{sec:1}

Consider the following kinetic equation for the one-particle distribution function $f(x,v,t)$,
\begin{equation}
\label{BGK}
(\partial_t f+ v\cdot \nabla_x f) (x,v,t) = \lambda(\varrho) \big(\varrho(x,t) M_f (x,v,t)-f(x,v,t)\big)\,,
\end{equation}
where
\[
M_f(x,v,t) = \frac 1 {(2 \pi T(x,t))^{3/2} } \exp\left(-\frac {|v-u(x,t)|^2}{2T(x,t)}\right)\,,
\]
\[
\begin{split}
& \varrho (x,t) = \int\! \rmd v\, f(x,v,t)\,, \quad \varrho u (x,t) = \int\! \rmd v\, f(x,v,t) v\,, \\ & \varrho(|u|^2 +3T)(x,t)  = \int\! \rmd v\, f(x,v,t) |v|^2 \,,
\end{split}
\]
and $\lambda$ is a suitable positive function of the spatial  density. Here, $(x,v)$ denotes position and velocity of the particle, respectively, and $t>0$ is the time. 

The evolution equation \eqref{BGK} describes the behavior of a particle moving freely. In addition, the particle thermalizes instantaneously at a random time of intensity $\lambda >0$. The Maxwellian $M_f$ has mean velocity and temperature given by $f$ itself.

Such a model was introduced by P.L.~Bhatnagar, E.P.~Gross, and M.~Krook in \cite {BGK} to deal with situations when the mean free path of the particle system is very small, but the hydrodynamical regime is not yet appropriate.

In the original paper \cite{BGK} the authors consider $\lambda(\varrho)=\varrho$. This choice, even if natural, presents serious problems from the mathematical side. Indeed, up to now there exists a constructive existence and uniqueness result only when $\lambda=1$ \cite {PP}, which is the case treated here, or when $\lambda$ is a smooth bounded function, case in which the analysis of \cite{PP} can be easily extended. 

Clearly, the BGK model must preserve mass, momentum, energy, and satisfy the $H$-Theorem. It also exhibits the usual hydrodynamic behavior whenever $\lambda \to \infty$. The interest of the model is related to the fact that the instantaneous thermalization described by \eqref{BGK} is much easier to compute compared to a huge amount of collisions which however produces the same effect at the end.

We now try to present an heuristic derivation of the BGK equation, which is indeed not a simple arbitrary toy model but it is based on reasonable physical arguments.

The starting point is the usual Boltzmann equation,
\begin{equation}
\label{B1}
(\partial_t +v \cdot \nabla_x ) f =\frac 1{\eps} Q(f,f)\,,
\end{equation}
where $Q$ is the collision operator, which we do not make explicit here, and $\eps$ is a  very small scale parameter. Fixed $t>0$, we represent the solution $f(t)$ to \eqref{B1} with initial datum $f_0$ in terms of the Trotter product formula: letting $n=\lfloor t/\tau \rfloor$ with $\tau >0$ very small,
\begin{equation}
\label{T}
f(t) \approx f(n\tau)  \approx (S_0 (\tau) S_h(\tau))^nf_0\,,
\end{equation}
where $S_0(t)f(x,v)=f(x-vt,v)$ is the free stream operator and $S_h(t) f_0$ is the solution to the homogeneous Boltzmann equation
\[
\eps \partial_t  f = Q(f,f)
\]
with initial datum $f_0$.

By virtue of the well known properties of the homogeneous Boltzmann equation,
\[
\lim_{\eps \to 0} S_0(\tau) f = \varrho M_f\,,
\]
so that we can replace in Eq.~\eqref{T} the term $S_h(\tau)$ by the transition probability $P$ such that $f \to \varrho M_f$ with probability $\tau <1$ and $f \to f$ with probability $1-\tau $. Thus
\[
f(n\tau) \approx (S_0(\tau) P)^nf_0\,,
\]
i.e.,
\begin{align}
\label{TD}
f(n\tau) & \approx S_0(\tau)  (S_0(\tau) P)^{n-1}f_0  + S_0(\tau) (P-1) (S_0(\tau) P)^{n-1}f_0  \nonumber \\ & = \cdots\cdots\cdots \nonumber \\ & = S_0(n\tau)f_0 + \sum_{k=1}^n S_0(k\tau) (P-1) (S_0(\tau) P)^{n-k}f_0 \nonumber \\ & \approx S_0(n\tau)f_0 + \sum_{k=1}^n S_0(k\tau) (P-1)  f[ (n-k) \tau]\,.
\end{align}
But 
\[
\frac 1\tau (P-1)f = \varrho M_f- f\,,
\]
thus in the limit $\tau\to 0$ Eq.~\eqref{TD} reads
\[
f(t) = S_0(t)f_0  + \int_0^t \! \rmd s\, S_0(s)  ( \varrho M_f-f) (t-s)\,,
\]
or, equivalently, \eqref{BGK} with  $\lambda=1$.

We assume $\lambda=1$ however, when the density is high, the transition probability of the event $f \to \varrho M_f$ should  increase so that we can recover, by the above heuristic argument, the BGK equation with $\lambda=\varrho$ as well, as originally proposed in \cite{BGK}.

The heuristic argument leading to the BGK picture starts from the Boltzmann equation and applies in situations close to the hydrodynamical regime, but still with a finite mean free path. Thus, this equation does not seem to be consequence of a scaling limit for which one obtains either the Boltzmann equation in the  low density (or Boltzmann-Grad) limit, or the hydrodynamical equations in a mere space-time scaling.
 
In the present paper we introduce a stochastic system of $N$ interacting particles on a $d$-dimensional torus ($d=2,3$), and prove the convergence of the one-particle law to the corresponding solution to the BGK equation in the limit $N \to \infty$ . In doing this, we follow a previous similar approach developed in \cite{BHP}. The improvements we present here are related to the following aspect. In the particle model we introduce a function $x \to \varphi(x)$ ($x$ is the space variable) which we use to compute the local empirical hydrodynamical fields,
\[
\begin{split}
& \varrho (x) = \frac 1N \sum_{j=1}^N \varphi (x-x_j)\,, \quad \varrho u (x) = \frac 1N \sum_{j=1}^N \varphi (x-x_j)v_j\,, \\ & \varrho(|u|^2 +3T)(x) = \frac 1N \sum_{j=1}^N \varphi (x-x_j)v_j^2\,,
\end{split}
\]
where $(x_1 \cdots x_N; v_1 \cdots v_N)$ are the sequence of positions and velocities of the particle configuration. In \cite{BHP} the function $\varphi$ was assumed to be strictly positive to avoid divergences due to the possibility of a local vacuum. This is physically not very reasonable because long-range effects should not play any role in the jump mechanism of a tagged particle. Here, we remove such an hypothesis allowing $\varphi$ to be compactly supported. In other words, $\varphi$ can be thought as a smoothed version of the characteristic function of a small ball.

We remark  that the  analysis of the present paper is, in a sense, equivalent to the introduction of the stochastic particle system  which  is the inhomogeneous version of the well known Kac's model \cite{PWZ,CPW}, and it is the conceptual basis of the  usual Monte Carlo Direct Simulation Method (Bird's scheme), see, e.g., \cite{CIP}, to approximate the solutions of the usual Boltzmann equation. 

Here, we work in a canonical context, i.e., the number of particles $N$ is fixed. In \cite{MW} the authors derive a linear version of the homogeneous BGK equation, starting from a suitable two species particle system in the microcanonical setting. Namely, the energy of the system they consider is also fixed. This is more related  to the spirit of the original Kac's model. 

The plan of the paper is the following. We first fix the cutoff function $\varphi$ and a regularized version of the BGK equation (see Eq.~\eqref{eq:kin} below) in which the hydrodynamical fields are smeared. After fixing notation and establishing the main result in Section \ref{sec:2}, we introduce in Section \ref{sec:3} a coupling between the particle system and $N$ independent copies of a one-particle stochastic process associated to the regularized BGK equation, which is the basic tool for the proof of convergence. Sections \ref{sec:4} and \ref{sec:5} contain the preliminary lemmas and the key result on the proximity of the particle system to the regularized equation. In Section \ref{sec:6} we remove the cutoff as in \cite{BHP}. But here we choose a different method, by working not on the equations but on the processes, again with a coupling technique. As matter of facts, the convergence of the laws is obtained in the 2-Wasserstein distance, hence weaker with respect to the result in \cite{BHP} holding in a weighted $L^1$ space, but in addition here we also prove the convergence of the processes, and this is closer to the spirit of the present analysis. The final result follows by a diagonal limit. 

\section{Notation and results}
\label{sec:2}

For $d=2,3$ we let 
\begin{itemize}
\item $\bb T^d = \big(\bb R/(\frac 12 + \bb Z)\big)^d$ be the $d$-dimensional torus of side length one. 
\item $\varphi\in C^\infty(\bb T^d;\bb R_+)$ be a smearing function such that\footnote{The assumption $\varphi_0\ge 1$ in unnecessary but it makes cleaner some estimates. On the other hand, it is not restrictive as we are interested in the case when $\varphi$ converges to the $\delta$-function.}
\begin{equation}
\label{eq:varphi}
\begin{split}
& \varphi_0 := \varphi(0) = \max\varphi \ge 1 \,, \quad \varphi(x) = 0 \text{ for } |x| \ge \frac 12\,, \\  & \varphi(x) = \varphi(-x)\,, \quad \int\!\rmd y\, \varphi(y) = 1\,,
\end{split}
\end{equation}
where for $x\in \bb T^d$ we denote by $|x|$ its distance from $0$ (on the torus).   
\item $M_{u,T}=M_{u,T}(v)$, $v\in\bb R^d$, be the normalized Maxwellian density of mean velocity $u\in \bb R^d$ and temperature $T$, i.e.,
\begin{equation}
\label{max}
M_{u,T}(v) = \frac1{(2\pi T)^{d/2}} \exp\left(-\frac{|v-u|^2}{2T}\right).
\end{equation}
\end{itemize}
Note that
\[
u = \int\!\rmd v\, M_{u,T}(v)\, v \,, \qquad T = \frac 1d \int\!\rmd v\, M_{u,T}(v)\, |v-u|^2\,.
\]

\subsection{The BGK equation}
\label{sec:2.1}

We denote by $f = f(t) = f(x,v,t)$, where $(x,v)\in \bb T^d \times \bb R^d$ and $t\in \bb R_+$ is the time, the solution to the BGK equation,
\begin{equation}
\label{eq:bgk}
\partial_t f + v\cdot \nabla_x f = \varrho_f M_f - f\,,
\end{equation}
where $\varrho_f= \varrho_f(x,t)$ is the local density defined by
\begin{equation}
\label{eq:rho}
\varrho_f(x,t) = \int\!\rmd v\, f(x,v,t)\,,
\end{equation}
while $M_f = M_f(x,v,t)$ is the (local) Maxwellian given by
\begin{equation}
\label{eq:maxf}
M_f(x,v,t) = M_{u_f(x,t),T_f(x,t)}(v)\,,
\end{equation}
with $u_f=u_f(x,t)$ and $T_f=T_f(x,t)$ the local velocity and temperature, 
\begin{align}
\label{eq:uf}
\varrho_f(x,t) u_f(x,t) & = \int\! \rmd v\, f(x,v,t)\, v\,, \\ \label{eq:tf} \varrho_f(x,t) T_f(x,t) & = \frac 1d \int\! \rmd v\, f(x,v,t)\, |v-u_f(x,t)|^2\,.
\end{align}

Well-posedness of the BGK equation together with $L^\infty$ estimates for the hydrodynamical fields can be found in \cite{PP}. In particular, we consider as initial condition a probability density $f_0$ on $\bb T^d\times \bb R^d$ such that there are a function $a\in C(\bb R^d)$ and positive constants $C_1,\alpha>0$ such that
\begin{equation}
\label{eq:f0}
\begin{split}
& a(v) \le f_0(x,v) \le C_1\rme^{-\alpha|v|^2} \quad \forall\, (x,v)\in \bb T^d\times \bb R^d\,, \\ & a \ge 0\,, \quad C_2 := \int\!\rmd v \, a(v) >0\,.
\end{split}
\end{equation}
Therefore, from \cite[Theorem 3.1]{PP} the following proposition follows.

\begin{proposition}
\label{prop:bgk}
There exists a mild solution $f=f(t)=f(x,v,t)$ to Eq.~\eqref{eq:bgk} with initial condition $f(x,v,0) = f_0(x,v)$ satisfying Eq.~\eqref{eq:f0}.\footnote{This means that $f$ solves the integral equation,
\[
f(x,v,t) =  f_0(x-vt,v) + \int_0^t\!\rmd s\,(\varrho_f M_f-f) (x-v(t-s),v,s)\,,
\] 
which formally derives from Eq.~\eqref{eq:bgk} via Duhamel formula.} Moreover, there are a non-decreasing finite function $t\mapsto K_{q,t} = K_{q,t}(f_0)$, $q\in\bb N$, and a non-increasing positive function $t\mapsto A_t = A_t(f_0)$ such that, for any $(x,t)\in \bb T^d\times \bb R_+$,
\begin{align}
\label{eq:utf}
& |u_f(x,t)| + T_f(x,t) + \mc N_q(f(t)) \le K_{q,t}\,, \\ \label {eq:rT>} & \varrho_f(x,t) \ge C_2 \rme^{-t}\,, \quad T_f (x,t) \ge A_t \,,
\end{align}
where
\begin{equation}
\label{eq:Nn}
\mc N_q(f) := \sup_{(x,v)\in \bb T^d\times \bb R^d} f(x,v) (1+ |v|^q)\,.
\end{equation}
Finally, the above solution is unique in the class of functions $f=f(t)=f(x,v,t)$ such that, for some $q>d+2$, $\sup_{t\le \tau}\mc N_q(f(t)) < +\infty$ for any $\tau>0$.
\end{proposition}

\subsection{The stochastic particle system}
\label{sec:2.2}

We consider a system of $N$ particles confined in the torus $\bb T^d$. We denote by $Z_N=(X_N,V_N)$ the state of the system, where $X_N\in(\bb T^d)^N $ and $V_N\in (\bb R^d)^N $ are the positions and velocities of the particles, respectively. The particles move randomly, governed by the stochastic dynamics defined as below.

Setting $X_N=(x_1, \dots, x_N)$ and $V_N=(v_1, \dots, v_N)$, we introduce the smeared empirical hydrodynamical fields $\varrho_N^\varphi$, $u_N^\varphi$, and $T_N^\varphi$ (depending on $Z_N$) defined by 
\begin{equation}
\label{empf}
\begin{split}
& \varrho_N^\varphi(x) = \frac 1N \sum_{j=1}^N \varphi (x-x_j)\,, \quad  \varrho_N^\varphi u_N^\varphi(x) = \frac 1N \sum_{j=1}^N \varphi (x-x_j) v_j\,, \\ & \varrho_N^\varphi T_N^\varphi(x) = \frac 1{Nd} \sum_{j=1}^N \varphi (x-x_j) |v_j-u_N^\varphi(x)|^2 \,,
\end{split}
\end{equation}
with $\varphi$ as in Eq.~\eqref{eq:varphi}. The particle system evolves according to the Markovian stochastic dynamics whose generator $\mc L_N$ is defined as
\begin{align}
\label{gen}
\mc L_N{G}(Z_N) & = [(V_N\cdot\nabla_{X_N})G] (Z_N) +  \sum _{i=1}^N \int\! \rmd \tilde x_i \, \varphi(\tilde x_i - x_i) \nonumber \\ &  \quad \times \left[ \int\! \rmd \tilde v_i\, M_{Z_N}^\varphi (\tilde x_i,\tilde v_i) G(Z_N^{i,(\tilde x_i,\tilde v_i)}) - G(Z_N) \right]\,.
\end{align}
Above (for given $(y,w) \in \bb T^d\times \bb R^d$), $Z_N^{i,(y,w)}=(X_N^{i,y},V_N^{i,w})$ is the state obtained from $Z_N=(X_N,V_N)$ by replacing the position $x_i$ and velocity $v_i$ of the $i$-th particle by $y$ and $w$ respectively; ${G}$ is a test function on the state space, and $M_{Z_N}^\varphi(x,v) $ is the Maxwellian associated to the empirical fields, i.e.,
\[
M_{Z_N}^\varphi(x,v) = M_{u_N^\varphi(x),T_N^\varphi(x)}(v)\,.
\]
Otherwise stated, the evolution $Z_N(t) = (X_N(t),V_N(t))$ is the Markov process in which at a random exponential time of intensity one the $i$-th particle performs a jump from its state $(x_i,v_i)$ to a new one $(\tilde x_i,\tilde v_i)$, extracted according to the distribution $\varphi(\cdot-x_i)$ for the position and then to the empirical Maxwellian $M_{Z_N}^\varphi (\tilde x_i,\cdot)$ for the velocity. 

The stochastic evolution is well posed since $\varrho_N^\varphi(\tilde x_i) \ge N^{-1} \varphi(\tilde x_i - x_i)$, so that  the Maxwellian $M_{Z_N}^\varphi (\tilde x_i,\tilde v_i)$ is well defined when $\varphi(\tilde x_i - x_i)>0$, and the integration in the right-hand side of Eq.~\eqref{gen} makes sense. On the other hand, the smeared hydrodynamical temperature $T_N^\varphi(\tilde x_i)$ may vanish, in which case we replace the Maxwellian $M_{Z_N}^\varphi (\tilde x_i,\tilde v_i)$ by a Dirac mass in $u_N^\varphi(\tilde x_i)$. In particular, this happens in the special case $\varrho_N^\varphi(\tilde x_i) = N^{-1} \varphi(\tilde x_i - x_i)$, which implies  $u_N^\varphi(\tilde x_i)=v_i$, so that $M_{Z_N}^\varphi (\tilde x_i,\tilde v_i) = \delta(\tilde v_i - v_i)$ (the velocity does not jump).

In the sequel, we will denote by $F_N(t) = F_N(Z_N,t)$ the density of the law of $Z_N(t)$ (but we will often refer to it as simply the law of the process). 

\subsection{The regularized BGK equation}
\label{sec:2.3}

The kinetic limit of the particle system introduced in Section \ref{sec:2.2} will be shown to be governed by the following regularized version of Eq.~\eqref{eq:bgk}, 
\begin{equation}
\label{eq:kin}
\partial_t g + v\cdot \nabla_x g = \varrho_g^\varphi M_g^\varphi - g\,,
\end{equation}
for the unknown distribution function $g = g(t) =g(x,v,t)$, where $M_g^\varphi$ is the Maxwellian
\begin{equation}
\label{eq:maxg}
M_g^\varphi(x,v,t) = M_{u_g^\varphi(x,t),T_g^\varphi(x,t)}(v)\,,
\end{equation}
and the fields $\varrho_g^\varphi = \varrho_g^\varphi(x,t) $, $u_g^\varphi=u_g^\varphi(x,t)$, and $T_g^\varphi=T_g^\varphi(x,t)$ are given by
\begin{align}
\label{eq:rphi}
\varrho_g^\varphi(x,t) & = (\varphi*\varrho_g) (x,t) = \int\!\rmd y\, \varphi(x-y)\varrho_g(y,t), \\ \label{eq:uphi}
\varrho_g^\varphi(x,t) u_g^\varphi (x,t) & = \int\! \rmd y\, \rmd v\, \varphi(x-y) g(y,v,t)\, v\,, \\ \label{eq:tphi} \varrho_g^\varphi(x,t) T_g^\varphi (x,t) & = \frac 1d \int\! \rmd y\, \rmd v\, \varphi(x-y) g(y,v,t)\, |v-u_g^\varphi(x,t)|^2\,,
\end{align}
with
\begin{equation}
\label{eq:rhog}
\varrho_g(x,t) = \int\!\rmd v\, g(x,v,t)\,.
\end{equation}

The content of Proposition \ref{prop:bgk} extends to the regularized BGK equation, in particular the $L^\infty$ estimates do not depend on the smearing function $\varphi$. This is the matter of \cite[Proposition 2.2]{BHP} - which we report below for the convenience of the reader, noticing that it applies also in the present context since the proof does not depend on the assumption (done in \cite{BHP}) that $\varphi$ is strictly positive.   

\begin{proposition}
\label{prop:stim_uT}
Let $g=g(t)=g(x,v,t)$ be the solution to Eq.~\eqref{eq:kin} with initial condition $g(x,v,0) = f_0(x,v)$, $f_0$ as in Proposition \ref{prop:bgk}, i.e., satisfying Eq.~\eqref{eq:f0}. Then, similar estimates hold for the corresponding hydrodynamical fields, namely,
\begin{align}
\label{eq:utK}
& |u_g^\varphi(x,t)| + T_g^\varphi(x,t) + \mc N_q(g(t)) \le K_{q,t}\,, \\ \label{stimrho} & \varrho_g(x,t) \ge C_2 \rme^{-t}\,, \quad \varrho_g^\varphi(x,t) \ge C_2\rme^{-t}\,, \\ \label{eq:TA} & T_g^\varphi (x,t) \ge A_t\,,
\end{align}
(with $t\mapsto K_{q,t} = K_{q,t}(f_0)$, $q\in\bb N$ non-decreasing and $t\mapsto A_t = A_t(f_0)$ non-increasing, both positive and independent of $\varphi$). 
\end{proposition}

\subsection{Kinetic limit}
\label{sec:2.4}

We can now state the key result of the paper, concerning the kinetic limit of the stochastic particle system. 

\begin{theorem}
\label{teo:main}
Suppose that the law of $Z_N(0)$ has density $F_N(0) = f_0^{\otimes N}$, where $f_0$ satisfies the assumptions detailed in Eq.~\eqref{eq:f0}, and let $g=g(t)=g(x,v,t)$ be the solution to Eq.~\eqref{eq:kin} with initial condition $g(0)=f_0$. 

Let $f_j^N(t)$, $j\in\{1,\ldots, N\}$, be the $j$-particle marginal distribution function of the (symmetric) law $F_N(t)$,i.e.,
\[
f_j^N(x_1,\ldots,x_j,v_1,\ldots,v_j,t) = \int\! \rmd x_{j+1} \cdots \rmd x_N\, \rmd v_{j+1} \cdots \rmd v_N\, F_N(X_N,V_N,t)\,.
\]
Then, the 2-Wasserstein distance $\mc W_2\big(f_j^N(t), g(t)^{\otimes j}\big)$ vanishes as $N\to +\infty$ for any $j\in \bb N$ and $t \ge 0$. More precisely, for each $T>0$ there exists $L_T = L_T(f_0)$ such that, for any $j\in\{1,\ldots, N\}$,
\begin{equation}
\label{w2s}
\mc W_2\big(f_j^N(t),g(t)^{\otimes j}\big)^2 \le \frac{jL_T}{N^{1/4} } \exp (L_T\Gamma_\varphi) \quad \forall\, t\in [0,T] \quad \forall\, N>N_\varphi\,,
\end{equation}
where $\Gamma_\varphi$ and $N_\varphi$ are explicitly computable positive numbers depending solely on the smearing function $\varphi$ (see Eq.~\eqref{N0Gammaphi} below). In particular, the one particle marginal distribution function $f_1^N(t)$ converges weakly to $g(t)$ as $N\to +\infty$ for any $t \ge 0$.
\end{theorem}

\begin{remark}
\label{rem:wass}
We recall that if $\mu$ and $\nu$ are two probability measures on a metric space $(M,d)$ with finite second moment, the 2-Wasserstein distance between $\mu$ and $\nu$ is defined as 
\[
\mc W_2(\mu ,\nu) = \left(\inf_{\gamma\in \mc P(\mu,\nu)} \int_{M\times M}\!\rmd \gamma(x,x')\, d(x,x')^2\right)^{1/2},
\]
where $\mc P(\mu,\nu)$ denotes the collection of all the probability measures on $M\times M$ with marginals $\mu$ and $\nu$. In Theorem \ref{teo:main}, $M=(\bb T^d)^j\times (\bb R^d)^j$ and $\mc W_2\big(f_j^N(t), g(t)^{\otimes j}\big)$ denotes the 2-Wasserstein distance between the probability measures with densities $f_j^N(t)$ and $g(t)^{\otimes j}$, respectively.
\end{remark}

The proof of Theorem \ref{rem:wass} will be presented in Section \ref{sec:5} after some preliminaries in Sections \ref{sec:3} and \ref{sec:4}.

The convergence of the particle system to the true BGK equation Eq.~\eqref{eq:bgk} is now obtained through a rescaling of the smearing function $\varphi$, by setting
\begin{equation}
\label{scaphi}
\varphi (x) = \varphi^{(\eps)} (x) := \frac{1}{\eps^d}\bar \varphi\bigg(\frac{x}{\eps}\bigg)\,,
\end{equation}
where $\bar \varphi$ is fixed (it varies on the scale of order one) and satisfies \eqref {eq:varphi}. Clearly, in this case $\varphi_0 \approx \eps^{-d}$ and $\|\nabla \varphi\|_\infty \approx \eps^{-d-1}$. We have the following result.

\begin{theorem}
\label{teo:lim}
Let $\varphi=\varphi^{(\eps)}$ be as in Eq.~\eqref{scaphi}, suppose $f_0$ satisfies Eq.~\eqref{eq:f0} and in addition that, for some $q>d+2$,
\begin{equation}
\label{grad}
\mc N_q(|\nabla_x f_0|) < + \infty\,.
\end{equation}
Then, for each $T>0$ there exists $C_T = C_T(f_0)$ such that,
\begin{equation}
\label{conveps}
\mc W_2(f(t),g(t))^2 \le C_T\, \eps \quad \forall\, t\in [0,T]\,,
\end{equation} 
where $f(t)$ and $g(t)$ are the solutions to Eq.~\eqref{eq:bgk} and Eq.~\eqref{eq:kin} respectively, with same initial condition $f_0$. 
\end{theorem}

We are now in position to formulate the main result.

\begin{theorem}
\label{teo:main1}
Under the hypotheses of Theorem \ref{teo:main} and Theorem \ref{teo:lim}, suppose that $\eps$ vanishes gently when $N$ diverges (for instance $\eps=(\log N)^{-\mu}$ with $\mu$ sufficiently small). Then, for all integer positive $j$,
\begin{equation}
\label{conv}
\lim_{N\to \infty} \mc W_2 ( f^N_j(t), f(t) ^{\otimes j} ) = 0 \quad \forall\, t\in [0,T]\,. 
\end{equation}
\end{theorem}

Some comments are in order. Theorem \ref{teo:main1} is actually a corollary of Theorem \ref{teo:main} and Theorem \ref{teo:lim} via a triangular inequality. The short proof will be presented at the end of Section \ref{sec:6}.

The convergence expressed in Theorem \ref {teo:main1} is very slow. In particular, choosing $\eps=(\log N)^{-\mu}$ with $\mu$ sufficiently small we obtain
\[
\mc W_2(f^N_j(t), f(t) ^{\otimes j})^2 \le \mathrm{const} \; j (\log N)^{-\mu} \,.
\]
The condition on $\eps$ is due to the fact that (when $\varphi=\varphi^{(\eps)}$) Eq.~\eqref{w2s} holds with $\Gamma_\varphi  \approx \eps^{-a}$ and $N_\varphi \approx \eps^{-b}$ for suitable $a,b>1$, so that the condition $N>N_\varphi$ is satisfied and the divergence $\exp\big(C\eps^{-a}\big)$ appearing in the right-hand side is compensated by the term $N^{-1/4}$. We did not tried to optimize further our estimates since an effort in this direction would not improve so much the result. A similar feature is also present in \cite{BHP}, where the physically reasonable scaling is discussed in Section \ref{sec:5}.

\section{Reformulation of the problem}
\label{sec:3}

Following the strategy developed in \cite{BHP}, we prove Theorem~\ref{teo:main} by showing that the stochastic particle system is close (as $N\to +\infty$) to an auxiliary process, whose asymptotic as $N \to \infty$ is obvious.

\subsection{Coupling with an independent process}
\label{sec:3.1}

The auxiliary process is denoted by $\Sigma_N(t) = (Y_N(t),W_N(t)) \in (\bb T^d)^N\times (\bb R^d)^N$ and it is defined according to the following construction.

Let $g=g(t) = g(x,v,t)$ be as in Proposition \ref{prop:stim_uT}. Denote by $(x(t),v(t))\in \bb T^d\times \bb R^d$ the one-particle jump process whose generator is given by
\begin{equation}
\label{nonlin}
\mc L_1^g\psi(x,v) = [(v\cdot \nabla_x )\psi](x,v) + \int\! \rmd \tilde x \, \varphi(\tilde x - x) \left[\int\! \rmd \tilde v\, M_g^\varphi(\tilde x,\tilde v) \psi(\tilde x,\tilde v) - \psi(x,v) \right]\,,
\end{equation}
where $\psi$ is a test function and $M_g^\varphi$ is defined in Eq.~\eqref{eq:maxg}. In particular, if the initial distribution has a density, the same holds at any positive time and the probability density of $(x(t),v(t))$ solves the regularized BGK equation \eqref{eq:kin}. This kind of process is usually called non-linear since its generator is implicitly defined through the law of the process itself.

The process $\Sigma_N(t)$ is then defined by $N$ independent copies of the above process, i.e., as the Markov process on $(\bb T^d)^N\times (\bb R^d)^N$ with generator
\begin{align}
\label{genf}
\mc L_N^g{G}(Z_N) & = [(V_N\cdot\nabla_{X_N} ){G}] (Z_N) + \sum _{i=1}^N \int\! \rmd \tilde x_i \, \varphi(\tilde x_i - x_i) \nonumber \\ & \quad \times \left[ \int\!\rmd \tilde v_i\, M_g^\varphi (\tilde x_i,\tilde v_i) G(Z_N^{i,(\tilde x_i,\tilde v_i)}) - G(Z_N) \right] \,.
\end{align}
Note that the only difference with respect to Eq.~\eqref{gen} is the replacement of $M_{Z_N}^\varphi$ by $M^\varphi_g$.

As in \cite{BHP}, the closeness of $Z_N(t)$ and $\Sigma_N(t)$ is proved by introducing a suitable coupled process $Q_N(t) = (Z_N(t),\Sigma_N(t))$. More precisely, the coupled process is the Markov process whose generator $\mc L_Q$ is defined in the following way. We let $Z_N=(X_N,V_N)$, $\Sigma_N=(Y_N,W_N)$, with $X_N=(x_1, \dots, x_N)$, $V_N=(v_1, \dots, v_N)$, $Y_N=(y_1, \dots, y_N)$, and $W_N=(w_1, \dots, w_N)$. Then, for any test function $G=G(Z_N,\Sigma_N)$,
\begin{align}
\label{genQ}
& \mc L_Q G(Z_N,\Sigma_N) = [(V_N\cdot\nabla_{X_N} + W_N\cdot\nabla_{Y_N}) G] (Z_N,\Sigma_N) \nonumber +\int\! \rmd \tilde x_i\, \rmd \tilde y_i \, \Phi_{x_i,y_i}(\tilde x_i,\tilde y_i) \nonumber \\ & \qquad \times \left[ \int\! \rmd \tilde v_i\, \rmd \tilde w_i \,\mc M^\varphi (\tilde x_i,\tilde v_i;\tilde y_i,\tilde w_i) {G} (Z_N^{i,(\tilde x_i,\tilde v_i)},\Sigma_N^{i,(\tilde y_i,\tilde w_i)}) - G(Z_N,\Sigma_N) \right]\,,
\end{align}
In Eq.~\eqref{genQ}, for given $\tilde x,\tilde y\in \bb T^d$, $\mc M^\varphi(\tilde x,v;\tilde y,w)$ is the joint representation of the Maxwellians $M_{Z_N}^\varphi(\tilde x,v)$ and $M_g^\varphi(\tilde y,w)$ that realizes the 2-Wasserstein distance between the marginals, whose square is given by (see, e.g., \cite{OP})
\begin{equation}
\label{w2m}
\mc W_2\big(M_{Z_N}^\varphi(x,\cdot),M_g^\varphi(y,\cdot)\big)^2 = |u_N^\varphi(x)-u_g^\varphi(y)|^2 + d\Big(\sqrt{T_N^\varphi(x)} - \sqrt{T_g^\varphi(y)}\Big)^2\,.
\end{equation}
While, for given $x,y\in \bb T^d$, $\Phi_{x,y}(\tilde x,\tilde y)$ is the joint representation of the probability densities $\varphi_x(\tilde x) = \varphi(\tilde x-x)$ and  $\varphi_y(\tilde y) = \varphi(\tilde y-y)$ defined as
\begin{equation}
\label{jointxy0}
\Phi_{x,y}(\tilde x,\tilde y) = \varphi_x(\tilde x) \delta(\tilde x -x - \tilde y +y)\,,
\end{equation}
where $\delta(z)$ denotes the Dirac measure on $\bb T^d$ centered in $z=0$. In particular, for any integrable function $J$ on $\bb T^d$,
\begin{equation}
\label{jointxy}
\int\! \rmd\tilde x\, \rmd\tilde y\, \Phi_{x,y}(\tilde x,\tilde y) J(\tilde x-\tilde y) = J(x-y)\,.
\end{equation}

In words, the coupling is given by the Markov process in which at a random exponential time of intensity one, the $i$-th pair of particles makes the jump from $(x_i,v_i,y_i,w_i)$ to $(\tilde x_i,\tilde v_i, \tilde y_i, \tilde w_i) = (x_i+\xi,\tilde v_i, y_i+\xi, \tilde w_i)$, where $\xi$ is distributed according to $\varphi$, and $(\tilde v_i, \tilde w_i)$ according to the prescribed joint representation of $M^\varphi_{Z_N}(\tilde x_i, \cdot)$ and $M^\varphi_g(\tilde y_i,\cdot)$. We denote by $\rmd R_N(t) = \rmd R_N(Z_N,\Sigma_N,t)$ the law of $Q_N(t)$ and assume that, initially,
\[
\rmd R_N(0) = \delta(X_N-Y_N) \delta(V_N-W_N) f_0^{\otimes N}(X_N,V_N) \,\rmd Z_N\, \rmd \Sigma_N\,.
\]
In particular, recalling the notation introduced in Remark \ref{rem:wass},  
\[
\rmd R_N(t) \in \mc P\big(F_N(t)\rmd Z_N, g(t)^{\otimes N}\rmd \Sigma_N\big)\,.
\]

\subsection{Estimating the distance between the processes}
\label{sec:3.2}

We adopt the same strategy of \cite{BHP} and introduce the quantity
\[
I_N(t) := \int\! \rmd R_N(t)\, (|x_1-y_1|^2 + |v_1-w_1 |^2)\,.
\]
As $\rmd R_N(t)$ is symmetric with respect to particle permutations we have 
\[
I_N(t) = \frac 1j \int\! \rmd R_N(t)\, \sum_{i=1}^j (|x_i-y_i|^2 + |v_i-w_i |^2) \quad \forall\, j\in\{1,\ldots,N\}\,,
\] 
so that 
\[
\mc W_2\big(f_j^N(t),g(t)^{\otimes j}\big) \le \sqrt{j I_N(t)}  \quad \forall\, j\in\{1,\ldots,N\}\,,
\]
by the definition of the 2-Wasserstein distance. Therefore, the proof of Theorem \ref{teo:main} reduces to show that for each $T>0$ there exists $L_T = L_T(f_0)$ such that,
\begin{equation}
\label{I_N}
I_N(t) \le \frac{L_T}{N^{1/4} } \exp (L_T\Gamma_\varphi) \quad \forall\, t\in [0,T] \quad \forall\, N>N_\varphi\,,
\end{equation}
for suitable $\Gamma_\varphi$ and $N_\varphi$. To this end, we compute
\begin{align*}
\dot I_N(t)  & =  \int\! \rmd R_N(t)\, \mc L_Q (|x_1-y_1|^2 + |v_1-w_1 |^2) \\ & =  \int\! \rmd R_N(t)\, (v_1 \cdot \nabla_{x_1} + w_1 \cdot \nabla_{y_1})  |x_1-y_1|^2 \\ & \quad - N \int\! \rmd R_N(t)\,  ( |x_1-y_1|^2 + |v_1-w_1|^2) \\ &  \quad + \sum_{i=2}^N \int\! \rmd R_N(t)\, (|x_1-y_1|^2 + |v_1-w_1|^2) + \int\! \rmd R_N(t)\, |x_1-y_1|^2 \\ & \quad + \int\! \rmd R_N(t) \int\! \rmd \xi\, \varphi(\xi) \int\! \rmd \tilde v_1\, \rmd \tilde w_1\, \mc M^\varphi(x_1+\xi, \tilde v_1; y_1+\xi, \tilde w_1) |\tilde v_1- \tilde w_1|^2\,,
\end{align*}
where the first two terms in the right-hand side arise from the stream part ($V_N\cdot\nabla_{X_N}G + W_N\cdot\nabla_{Y_N}G$) and the loss part ($-NG$) of the generator $\mc L_Q$, respectively. We note that the loss term is partially compensated by the third term, while the stream part is equal to
\[
2 \int\! \rmd R_N(t)\, (v_1-w_1) \cdot (x_1-y_1) \le \int\! \rmd R_N(t)\, (|x_1-y_1|^2+|v_1-w_1 |^2)\,,
\]
where, with an abuse of notation, in the left-hand side we denote by $(x_1-y_1)$  a vector $\eta \in \bb R^d$ in the equivalence class defined by $x_1-y_1\in \bb T^d= \big(\bb R/(\frac 12 + \bb Z)\big)^d$ with $|\eta|=|x_1-y_1|$ and, when not uniquely determined by these conditions, with the minimum value of $(v_1-w_1)\cdot \eta$ (however, this is an event of vanishing measure and will not play any role in the sequel). Finally, the last term is given by Eq.~\eqref{w2m}. Therefore
\begin{equation}
\label{I<1}
\dot I_N(t) \le I_N(t) + \int\! \rmd R_N(t)\, D(Z_N,\Sigma_N)\,,
\end{equation}
with
\begin{align}
\label{w2}
D(Z_N,\Sigma_N) & = \int\! \rmd \xi\, \varphi(\xi) |u_N^\varphi(x_1+\xi)-u_g^\varphi(y_1+\xi)|^2 \nonumber \\ & \quad + \int\! \rmd \xi\, \varphi(\xi)\, d\Big(\sqrt{T_N^\varphi(x_1+\xi)} - \sqrt{T_g^\varphi(y_1+\xi)}\Big)^2.
\end{align}
Our goal is to estimate $\int\! \rmd R_N(t)\, D(Z_N,\Sigma_N)$ from above with a constant (independent of $N$) multiple of $I_N(t)$ plus a small term of order $1/N^{1/4}$, so that Eq.~\eqref{I_N} follows from Gr\"{o}nwall's inequality.

As noticed in \cite{BHP}, in estimating $D(Z_N,\Sigma_N)$ it is useful to replace $\varrho_g^\varphi$, $u_g^\varphi$, $T_g^\varphi$ by
\begin{equation}
\label{empf2}
\begin{split}
& \tilde \varrho_N^\varphi(x) = \frac 1N \sum_{j=1}^N \varphi (x-y_j)\,, \qquad  \tilde \varrho_N^\varphi \tilde u_N^\varphi(x) = \frac 1N \sum_{j=1}^N \varphi (x-y_j) w_j\,, \\ & \tilde \varrho_N^\varphi \tilde T_N^\varphi(x) = \frac 1{Nd} \sum_{j=1}^N \varphi (x-y_j) |w_j - \tilde u_N^\varphi(x)|^2 \,,
\end{split}
\end{equation}
i.e., the empirical fields constructed via the variables $Y_N=(y_1, \dots, y_N) $ and $W_N=(w_1, \dots, w_N)$, distributed independently according to $g(t)^{\otimes N}$. By the law of large numbers, this replacement is expected to be small for large $N$. 

In the present case, the function $\varphi$ has compact support, so that there are particle configurations for which the smeared empirical densities defined in Eqs.~\eqref{empf} and \eqref{empf2} assume very small values (order $1/N$). This makes impossible to obtain (as in \cite{BHP}) a point-wise estimate of $D(Z_N,\Sigma_N)$. To overcome this difficulty, we decompose the phase space as the union of a ``good set'' $\mc G$, which will be defined in the next section, and its complement, the ``bad set'' $\mc G^\complement$. Roughly speaking, in the set $\mc G$, $D(Z_N,\Sigma_N)$ can be controlled similarly to what done in \cite{BHP}, while the contribution to $\dot I_N(t)$ coming from the integration on $\mc G^c$ will be treated by suitable probability estimates (actually, the decomposition of $\dot I_N(t)$ is more involved, as explained at the beginning of Section \ref{sec:5}).

\smallskip
\noindent\textbf{A notation warning.} In what follows, we shall denote by $C$ a generic positive constant whose numerical value may change from line to line and it may possibly depend on the fixed time $T$ and the initial condition $f_0$.

Furthermore, we will use both the notations $\id_B$ and $\id(B)$ to denote the characteristic function of the set $B$. We shall also use the shorten notation $\rmd g(t)^{\otimes N}$ to denote integration with respect to $ \rmd \Sigma_N\, g(t)^{\otimes N}$.

\section{Preliminary estimates}
\label{sec:4}

Recalling the assumptions Eq.~\eqref{eq:varphi} on $\varphi$, we fix $r\in (0,\frac{1}{10})$, with $r^{-1}\in \bb N$ and such that
\begin{equation}
\label{r0}
\varphi(x) > \frac{\varphi_0}2 \quad \forall\, x\in [-5r,5r]^d\,,
\end{equation}

Denote by $\{\Delta\}$ a partition of $\bb T^d$ into square boxes of side $r$. As a consequence, we have the following lower bound on the empirical densities,
\begin{equation}
\label{Ny}
N \tilde \varrho_N^\varphi(x) \ge \frac{\varphi_0}2 N_\Delta^Y  \quad \text{ if } x\in \Delta\,,
\end{equation}
where $N_\Delta^Y$ is the number of particles of the configuration $Y_N$ contained in $\Delta$.

\begin{lemma}
\label{lem:ny}
Given $T>0$ there is $A>0$ (depending only on $T$ and initial condition $f_0$) such that if
\begin{equation}
\label{ba}
\mc B_A := \{(Z_N,\Sigma_N) \colon \tilde \varrho_N^\varphi(x) > Ar^d\varphi_0 \;\; \forall\, x\in \bb T^d\}
\end{equation}
then
\begin{equation}
\label{eq:densy}
\int\! \rmd R_N(t) \id_{\mc B_A^\complement} \le \frac{C}{r^{3d}N} \quad \forall\, t\in [0,T]\,.
\end{equation}
\end{lemma}

\begin{proof}
By  Eq.~\eqref{Ny},
\[
\begin{split}
\int\! \rmd R_N(t) \id_{\mc B_A^\complement}  &= \int\! \rmd g(t)^{\otimes N}\,\id\big( \{Y_N\colon \exists\, x \in \bb T^d \text{ s.t. }\tilde \varrho_N^\varphi(x) \le Ar^d\varphi_0\}\big) \\ & \le \int\! \rmd g(t)^{\otimes N}\,  \id\big(\{Y_N\colon  \exists\,  \Delta \text{ s.t. } N_\Delta^Y \le 2Ar^dN\}\big)  \\ & \le \sum_{\{\Delta\}} \int\! \rmd  g(t)^{\otimes N}\,  \id_{N_\Delta^Y \le 2Ar^dN} \le \frac{1}{r^d} \max_{\Delta} \int\! \rmd g(t)^{\otimes N}\,  \id_{N_\Delta^Y \le 2Ar^dN}\,.
\end{split}
\]
We observe that $N_\Delta^Y = N \xi_N$ with $\xi_N = N^{-1} \sum_{j=1}^N \id_{y_j\in \Delta}$ the arithmetic mean of $N$ i.i.d.~random variables whose common expected value is
\[
 \bb E\xi_N = \bb E \id_{y_1\in \Delta} =  \int_\Delta\!\rmd y\, \varrho_g(x,t) \ge  C_2\rme^{-T}r^d \qquad \forall\,t\in [0,T]\,,
\]
having used Eq.~\eqref{stimrho} in the last inequality. We then choose $A=C_2\rme^{-T}/4$, whence 
\[
\id_{N_\Delta^Y \le 2Ar^dN} =  \id_{\xi_N \le 2Ar^d} \le \id_{|\xi_N- \bb E\xi_N| \ge \bb E\xi_N/2} \le  \id_{|\xi_N- \bb E\xi_N| \ge C_2\rme^{-T}r^dN/4}\,.
\]
Therefore, by Chebyshev's inequality,
\[
\int\! \rmd \Sigma_N\,  g(t)^{\otimes N}\,  \id_{N_\Delta^Y \le 2Ar^dN}  \le  \frac{16\rme^{2T}}{C_2^2 r^{2d}N} \bb E(\id_{y_1\in \Delta}- \bb E\id_{y_1\in \Delta})^2 \le \frac{C}{r^{2d}N}\,.
\]
Eq.~\eqref{eq:densy} is thus proved.
\end{proof}

\begin{lemma}[The good set]
\label{lem:nx}
Given $A>0$ as in Lemma \ref{lem:ny}, we let
\begin{equation}
\label{afi}
A_\varphi = \frac{Ar^d\varphi_0}{2\|\nabla\varphi\|_\infty}
\end{equation}
and define
\begin{equation}
\label{g}
\mc G := \mc G_1 \cap \mc B_A \;\text{ with }\; \mc G_1 := \bigg\{(Z_N,\Sigma_N) \colon \frac 1N\sum_{j=1}^N |x_j-y_j| \le  A_\varphi \bigg\}\,.
\end{equation}
Then
\begin{equation}
\label{stimdens}
\varrho_N^\varphi(x) > \frac{Ar^d\varphi_0}2\,, \quad \tilde \varrho_N^\varphi(x) > Ar^d\varphi_0 \quad \forall\, x\in \bb T^d \quad \text{ in the set }\mc G.
\end{equation}
\end{lemma}

\begin{proof}
The lower bound on $\tilde \varrho_N^\varphi(x)$ follows trivially from the definition of $\mc B_A$. Concerning the other bound, in the set $\mc G$ we have
\[
\begin{split}
\varrho_N^\varphi(x)  & \ge  \tilde \varrho_N^\varphi(x) - |\varrho_N^\varphi(x) - \tilde \varrho_N^\varphi(x)| \ge Ar^d\varphi_0 - \frac{\|\nabla\varphi\|_\infty}N \sum_{j=1}^N |x_j-y_j|\\ & \ge Ar^d\varphi_0 - \|\nabla\varphi\|_\infty A_\varphi = \frac{Ar^d\varphi_0}2\,,
\end{split}
\]
and the lemma is proved.
\end{proof}

\begin{lemma}
\label{lem:stimpij}
Define
\begin{align}
\label{pij}
p_{i,j} & = p_{i,j}(\xi) := \frac{\varphi(x_i+\xi-x_j)}{\sum_k \varphi(x_i+\xi-x_k)} = \frac{\varphi(x_i+\xi-x_j)}{N \varrho_N^\varphi(x_i+\xi)}\,, \\ \label{qij}
q_{i,j} & = q_{i,j}(\xi) := \frac{\varphi(y_i+\xi-y_j)}{\sum_k \varphi(y_i+\xi-y_k)} = \frac{\tilde \varphi(y_i+\xi-y_j)}{N \varrho_N^\varphi(y_i+\xi)}\,.
\end{align}
Then, recalling $\varphi_0=\max\varphi$,
\begin{align}
\label{pijstim}
& \sum_{j=1}^N p_{i,j} = 1\,, \qquad \int\!\rmd \xi\,\varphi(\xi)  \sum_{i=1}^N p_{i,j} \le \varphi_0\,, \\ \label{qijstim} &  \sum_{j=1}^N q_{i,j} = 1\,, \qquad \int\!\rmd \xi\,\varphi(\xi)  \sum_{i=1}^N q_{i,j} \le \varphi_0\,. 
\end{align}
\end{lemma}

\begin{proof}
The proofs of Eq.~\eqref{pijstim} and \eqref{qijstim} are the same, let us consider the first one. The normalization property $\sum_{j=1}^N p_{i,j} = 1$ is obvious, while (with the change of variable $\xi' = x_i + \xi$)
\[
\begin{split}
\int\!\rmd \xi\, \varphi(\xi) \sum_{i=1}^N p_{i,j} & \le \varphi_0 \int\!\rmd \xi\, \varphi(\xi) \sum_{i=1}^N \frac{1}{N \varrho_N^\varphi(x_i+\xi)} \\ & = \varphi_0 \sum_{i=1}^N \int\!\rmd \xi'\, \frac{\varphi(\xi'-x_i)}{N \varrho_N^\varphi(\xi')} = \varphi_0 \int\!\rmd \xi'\, \frac{N \varrho_N^\varphi(\xi')}{N \varrho_N^\varphi(\xi')} =\varphi_0 \,. 
\end{split}
\]
(recall the volume of the torus $\bb T^d$ is one).
\end{proof}

\begin{lemma}
\label{lem:llnw}
Given $T>0$, for each $p\in \bb N$ there is $M$ (depending only on $T$, $p$, and initial condition $f_0$) such that the following holds.

(1) For any  $j=1,\ldots,N$ we have
\begin{equation}
\label{eq:stimv}
\int\!\rmd R_N(t) |w_j|^p = \int\!\rmd g^{\otimes N}(t) |w_j|^p = \int\! \rmd g(t)\,  |w|^p \le \frac M2 \quad \forall\, t\in [0,T]\,.
\end{equation}

(2) If
\begin{equation}
\label{eq:llnw}
\mc G_{M,p} := \bigg\{(Z_N,\Sigma_N) \colon \frac 1N\sum_{j=1}^N |w_j|^p \le M \bigg\}
\end{equation}
then
\begin{equation}
\label{eq:stimllnw}
\int\!\rmd R_N(t) \id_{\mc G_{M,p}^\complement} \le \frac CN \quad \forall\, t\in [0,T]\,. 
\end{equation}
\end{lemma}

\begin{proof}
From the estimate on $\mc N_q(g)$ in Eq.~\eqref{eq:utK}, there is $M=M(T,p,f_0)$ such that $\int\!\rmd y\, \rmd w\, g(y,w,t) |w|^p \le M/2$ for any $t\in [0,T]$, which  proves Eq.~\eqref{eq:stimv}. Moreover, letting $\xi_N= \frac 1N\sum_j|w_j|^p$ and $\bb E(\xi_N) = \int\! \rmd g(t)^{\otimes N} \xi_N$, we have
\[
\int\!\rmd R_N(t) \id_{\mc G_{M,p}^\complement}  = \int_{{\sum}_j |w_j|^p > MN}\! \rmd g(t)^{\otimes N} \le \int_{|\xi_N-\bb E(\xi_N)|>M/2}\! \rmd g(t)^{\otimes N} \,,
\]
whence Eq.~\eqref{eq:stimllnw} follows from the law of large numbers (i.e., Chebyshev's inequality).
\end{proof}

\section{Proofs}
\label{sec:5}

We deduce an upper bound for the quantity $D(Z_N,\Sigma_N)$ introduced in Eq.~\eqref{w2}, which is the sum of several terms. To estimate the expectation of some of them, a separate analysis on the good set and its complement will be necessary. To this purpose, we first introduce the ``mixed temperature''
\[
\bar T^\varphi_N(x_1+\xi,y_1+\xi) = \frac 1d \sum_{j=1}^N p_{1,j} |w_j-\tilde u_N^\varphi(y_1+\xi)|^2\,.
\]

To simplify the notation, in what follows we will omit sometimes the explicit dependence on $x_1+\xi$ and $y_1+\xi$. By virtue of Eq.~\eqref{w2} we have
\begin{align}
\label{w2<<}
 D(Z_N,\Sigma_N) & \le \int\! \rmd \xi\, \varphi(\xi) \Big(2|u_N^\varphi - \tilde u_N^\varphi|^2 + 2|\tilde u_N^\varphi - u_g^\varphi|^2\Big)  \nonumber \\ & \quad+ \int\! \rmd \xi\, \varphi(\xi)\, \Big[2d\Big(\sqrt{T_N^\varphi} - \sqrt{\bar T_N^\varphi}\Big)^2  + 2d\Big(\sqrt{\bar T_N^\varphi} - \sqrt{T_g^\varphi}\Big)^2\Big],
\end{align}
where $\tilde u_N^\varphi$ is defined in Eq.~\eqref{empf2}. Recalling the definitions Eqs.~\eqref{pij} and \eqref{qij}, from Cauchy-Schwarz inequality,
\begin{equation}
\label{u-u}
|u_N^\varphi - \tilde u_N^\varphi|^2 \le 2\mc V+ 2 \bigg(\sum_{j=1}^N |p_{1,j}-q_{1,j}||w_j|\bigg)^2,
\end{equation}
where 
\begin{equation}
\label{DV}
\mc V = \mc V(\xi,Z_N,\Sigma_N)  := \sum_{j=1}^N p_{1,j} |v_j-w_j|^2\,.
\end{equation}
Concerning the difference between the empirical and mixed temperature, we observe that
\[
\begin{split}
\Big|T_N^\varphi - \bar T_N^\varphi\Big| & \le \frac 1d \sum_{j=1}^N p_{1,j}  \big| |v_j- u_N^\varphi|^2 - |w_j-\tilde u_N^\varphi|^2 \big| \\ & = \sum_{j=1}^N p_{1,j}  \Big|(v_j - u_N^\varphi - w_j + \tilde u_N^\varphi)\cdot (v_j - u_N^\varphi + w_j - \tilde u_N^\varphi) \Big|\\ & \le  \frac 1d  \sum_{j=1}^N p_{1,j} (|v_j -w_j| + |u_N^\varphi -\tilde u_N^\varphi|) |v_j-u_N^\varphi| \\ & \quad +  \frac 1d  \sum_{j=1}^N p_{1,j}  (|v_j -w_j| + |u_N^\varphi -\tilde u_N^\varphi|) |w_j-\tilde u_N^\varphi| \\ & \le \frac{1}{\sqrt d} \Big(\sqrt{\mc V} + |u_N^\varphi -\tilde u_N^\varphi| \Big) \Big(\sqrt{T_N^\varphi} +\sqrt{\bar T_N^\varphi}\Big)\,,
\end{split}
\]
where we used the Cauchy-Schwarz inequality in the last passage. Therefore, by Eq.~\eqref{u-u},
\begin{align}
\label{t-bt}
d \Big(\sqrt{T_N^\varphi} - \sqrt{\bar T_N^\varphi}\Big)^2 & = d \bigg(\frac{T_N^\varphi -  \bar T_N^\varphi}{\sqrt{T_N^\varphi} +\sqrt{ \bar T_N^\varphi}}\bigg)^2 \le 2 \mc V + 2 |u_N^\varphi -\tilde u_N^\varphi|^2\nonumber  \\ & \le 6 \mc V + 4 \bigg(\sum_{j=1}^N |p_{1,j}-q_{1,j}||w_j|\bigg)^2. 
\end{align}
On the other hand, from Eq.~\eqref{eq:TA},
\begin{equation}
\label{bt-gt}
d\Big(\sqrt{\bar T_N^\varphi} - \sqrt{T_g^\varphi}\Big)^2 =  d\bigg(\frac{\bar T_N^\varphi - T_g^\varphi}{\sqrt{T_N^\varphi} +\sqrt{T_g^\varphi}}\bigg)^2 \le \frac{2d\big(\bar T_N^\varphi - \tilde T_N^\varphi\big)^2}{A_t} + \frac{2d\big( \tilde T_N^\varphi - T_g^\varphi\big)^2}{A_t}\,,
\end{equation}
where $\tilde T_N^\varphi$ is defined in Eq.~\eqref{empf2}. Hence, by Eqs.~\eqref{w2<<}, \eqref{u-u}, \eqref{t-bt} and \eqref{bt-gt}, 
\begin{equation}
\label{w2<}
D(Z_N,\Sigma_N) \le D_1(Z_N,\Sigma_N) + D_2(Z_N, \Sigma_N) + \mc E(\Sigma_N) \,,
\end{equation}
with
\begin{align}
\label{D1}
D_1(Z_N,\Sigma_N) & = \int\! \rmd \xi\, \varphi(\xi) \, 16 \mc V(\xi,Z_N,\Sigma_N)\,, \\ \label{D2}  D_2(Z_N,\Sigma_N) & = \int\! \rmd \xi\, \varphi(\xi) \bigg[ 12 \bigg(\sum_{j=1}^N |p_{1,j}-q_{1,j}||w_j|\bigg)^2 + \frac{4d\big(\bar T_N^\varphi - \tilde T_N^\varphi\big)^2}{A_t} \bigg], \\ \label{E} \mc E(\Sigma_N) & =  \int\! \rmd \xi\, \varphi(\xi) \, 2 |\tilde u_N^\varphi- u_g^\varphi|^2 + \int\! \rmd \xi\, \varphi(\xi) \,  \frac{4d\big( \tilde T_N^\varphi - T_g^\varphi \big)^2 }{A_t}\,.
\end{align}
From Eqs.~\eqref{I<1}, \eqref{w2<}, and recalling the definition Eq.~\eqref{g} of the good set, we arrive at the following estimate on the derivative of $I_N(t)$, 
\begin{equation}
\label{I<2}
\dot I_N(t) \le I_N(t) + \mc D_a(t) + \mc D_b(t) + \mc D_c(t) + \mc D_d(t)\,,
\end{equation}
where
\[
\begin{split}
\mc D_a(t) & =  \int\! \rmd R_N(t)\, D_1(Z_N,\Sigma_N)\,, \qquad \mc D_b(t) = \int\! \rmd R_N(t)\, D_2(Z_N,\Sigma_N)  \id_{\mc G^\complement}\,, \\ \mc D_c(t) & = \int\! \rmd R_N(t)\, D_2(Z_N,\Sigma_N) \id_{\mc G}\,, \quad \mc D_d(t) = \int\! \rmd g(t)^{\otimes N} \, \mc E (\Sigma_N)\,.
\end{split}
\]

\subsection{Upper bound on  \texorpdfstring{$\mc D_a(t)$}{a}}
\label{sec:5.1}

Since $\rmd R_N(t)$ is symmetric with respect to particle permutations, 
\begin{align}
\label{da}
\mc D_a(t) & =  \int\! \rmd R_N(t) \int\! \rmd \xi\, \varphi(\xi) \, 16 \mc V(\xi,Z_N,\Sigma_N) \nonumber \\ & = \frac{16}N \sum_{i=1}^N  \int\! \rmd R_N(t)\int\! \rmd \xi\, \sum_{j=1}^N p_{i,j} |v_j-w_j|^2 \nonumber \\ & = \frac{16}N  \int\! \rmd R_N(t) \sum_{j=1}^N \bigg( \int\! \rmd \xi\, \varphi(\xi) \sum_{i=1}^N p_{i,j}\bigg) |v_j-w_j|^2 \nonumber \\ &  \le 16\varphi_0 I_N(t) \,,
\end{align}
where we used the upper bound of Eq.~\eqref{pijstim} in the last estimate.

\subsection{Upper bound on \texorpdfstring{$\mc D_b(t)$}{b}}
\label{sec:5.2}

By repeatedly applying the Cauchy-Schwarz inequality we have,
\[
\bigg(\sum_{j=1}^N |p_{1,j}-q_{1,j}||w_j|\bigg)^2 \le \bigg(\sum_{j=1}^N (p_{1,j} + q_{1,j})|w_j|\bigg)^2 \le 2  \sum_{j=1}^N (p_{1,j} + q_{1,j}) |w_j|^2 \,, 
\]
\[
\big(\bar T_N^\varphi - \tilde T_N^\varphi\big)^2 \le  \bigg(\sum_{j=1}^N (p_{1,j} + q_{1,j})|w_j-\tilde u_N^\varphi|^2\bigg)^2 \le C \sum_{j=1}^N (p_{1,j}+q_{1,j}) |w_j|^4 + C |\tilde u_N^\varphi\big|^4\,, 
\]
and
\[
|\tilde u_N^\varphi\big|^4 \le \bigg(\sum_{j=1}^N q_{1,j} |w_j| \bigg)^4 \le \sum_{j=1}^Nq_{1,j} |w_j|^4\,.
\]
Therefore, by Eq.~\eqref{D2} and recalling the definition Eq.~\eqref{g} of $\mc G$,
\begin{align}
\label{db1}
\mc D_b(t) & \le C  \int\! \rmd R_N(t)\, \int\! \rmd \xi\, \varphi(\xi) \sum_{j=1}^N (p_{1,j}+q_{1,j})(|w_j|^2 + |w_j|^4) \id_{\mc G^\complement} \nonumber \\ & \le C(\mc R_1 + \mc R_2)\,,
\end{align}
where
\[
\begin{split}
\mc R_1 & = \int\! \rmd R_N(t)  \int\! \rmd \xi\, \varphi(\xi) \sum_{j=1}^N (p_{1,j}+q_{1,j})(|w_j|^2 + |w_j|^4)  \id_{\mc B_A^\complement}\,, \\ \mc R_2 & = \int\! \rmd R_N(t)  \int\! \rmd \xi\, \varphi(\xi) \sum_{j=1}^N (p_{1,j}+q_{1,j})(|w_j|^2 + |w_j|^4) \id_{\mc G_1^\complement}\,.
\end{split}
\]

Since $\rmd R_N(t)$ and $\id_{\mc B_A^\complement}$ are symmetric with respect to particle permutations,
 \[
\begin{split}
\mc R_1 & = \frac 1N \sum_{i=1}^N \int\! \rmd R_N(t)  \int\! \rmd \xi\, \varphi(\xi) \sum_{j=1}^N (p_{i,j}+q_{i,j})(|w_j|^2 + |w_j|^4)  \id_{\mc B_A^\complement} \\ & =  \int\! \rmd R_N(t)\, \frac 1N \sum_{j=1}^N  \bigg( \int\! \rmd \xi\, \varphi(\xi) \sum_{i=1}^N (p_{i,j} + q_{i,j}) \bigg) (|w_j|^2 + |w_j|^4) \id_{\mc B_A^\complement}  \\ & \le 2\varphi_0  \int\! \rmd R_N(t)\, \frac 1N \sum_{j=1}^N (|w_j|^2 + |w_j|^4) \id_{\mc B_A^\complement}\,,
\end{split}
\]
where we used the upper bounds of Eqs.~\eqref{pijstim} and \eqref{qijstim} in the last inequality. Therefore, from the Cauchy-Schwarz inequality and  Eqs.~\eqref{eq:densy} and \eqref{eq:stimv},
\begin{align}
\label{db2}
\mc R_1 &  \le 2 \varphi_0 \bigg( \int\! \rmd R_N(t)\, \frac 1N \sum_{j=1}^N  (|w_j|^2 + |w_j|^4)^2 \bigg)^{1/2} \bigg( \int\! \rmd R_N(t)\, \id_{\mc B_A^\complement} \bigg)^{1/2} \nonumber \\ &\le \frac{C\varphi_0}{(r^{3d} N)^{1/2}} \bigg( \int\! \rmd R_N(t)\, \frac 1N \sum_{j=1}^N  (|w_j|^4 + |w_j|^8)\bigg)^{1/2} \le  \frac{C\varphi_0}{(r^{3d} N)^{1/2}}\,.
\end{align}

Analogously, since $\rmd R_N(t)$ and $\id_{\mc G_1^\complement}$ are symmetric with respect to particle permutations, by applying the upper bound of Eq.~\eqref{pijstim} and the Cauchy-Schwarz inequality,
\[
\begin{split}
\mc R_2 & = \frac 1N \sum_{i=1}^N \int\! \rmd R_N(t)  \int\! \rmd \xi\, \varphi(\xi) \sum_{j=1}^N (p_{i,j}+q_{i,j})(|w_j|^2 + |w_j|^4) \id_{\mc G_1^\complement} \\ & \le 2\varphi_0 \int\! \rmd R_N(t)  \frac 1N \sum_{j=1}^N (|w_j|^2 + |w_j|^4) \id_{\mc G_1^\complement} \\ & \le C \varphi_0 \int\! \rmd R_N(t)  \bigg(1 + \frac 1N \sum_{j=1}^N |w_j|^4\bigg) \id_{\mc G_1^\complement}\,.
\end{split}
\]
Recalling  Eq.~\eqref{eq:llnw}, we estimate $\id_{\mc G_1^\complement} \le  \id_{\mc G_1^\complement \cap \mc G_{M,4}} + \id_{\mc G_{M,4}^\complement}$ so that
\begin{align}
\label{db3}
\mc R_2 & \le C(1+M) \varphi_0\int\! \rmd R_N(t)\, \id_{\mc G_1^\complement} + C\varphi_0  \int\! \rmd R_N(t) \frac 1N \sum_{j=1}^N (1+|w_j|^4) \id_{\mc G_{M,4}^\complement} \nonumber \\ & \le \frac{C(1+M) \varphi_0}{A_\varphi^2} \int\! \rmd R_N(t)\, \frac 1N \sum_{i=1}^N  |x_i-y_i|^2 \nonumber \\ & \quad + C \varphi_0 \bigg[\int\! \rmd R_N(t)\, \bigg(\frac 1N \sum_{j=1}^N (1+|w_j|^4)\bigg)\bigg]^{1/2} \bigg( \int\! \rmd R_N(t)\, \id_{\mc G_{M,4}^\complement} \bigg)^{1/2}  \nonumber \\ & \le \frac{C (1+ M) \varphi_0}{A_\varphi^2} I_N(t) + \frac{C\varphi_0}{N^{1/2}}\,,
 \end{align}
where we used Chebyshev's inequality, Cauchy-Schwarz inequality twice, and finally Eqs.~\eqref{eq:stimv} and \eqref{eq:stimllnw}.

From Eqs.~\eqref{db1}, \eqref{db2}, and \eqref{db3}, and by Eq.~\eqref{afi}, we finally obtain
\begin{equation}
\label{da<}
\mc D_b(t) \le C \frac{\|\nabla\varphi\|_\infty^2}{r^{2d}\varphi_0} I_N(t) + \frac{C\varphi_0}{(r^{3d} N)^{1/2}}\,.
\end{equation} 

\subsection{Upper bound on \texorpdfstring{$\mc D_c(t)$}{c}}
\label{sec:5.3}

As
\[
\begin{split}
p_{1,j} - q_{1,j}  & = \frac{\varphi(x_1+\xi-x_j) - \varphi(y_1+\xi-y_j)}{N\varrho_N^\varphi(x_1+\xi)} \\ & \quad + \varphi(y_1+\xi-y_j) \frac{\sum_k [ \varphi(y_1+\xi-y_k) - \varphi(x_1+\xi-x_k)]}{N^2\varrho_N^\varphi(x_1+\xi) \tilde\varrho_N^\varphi(y_1+\xi)}\,,
\end{split} 
\]
from Eq.~\eqref{stimdens} we have that
\begin{align}
\label{stimp-q}
& |p_{1,j} - q_{1,j}| \le \frac{2\|\nabla\varphi\|_\infty}{Ar^d\varphi_0N} \big(|x_1-y_1| + |x_j-y_j| \big) \nonumber \\ & \qquad + \frac{2\varphi_0\|\nabla\varphi\|_\infty}{N^2(Ar^d\varphi_0)^2} \sum_{k=1}^N \big(|x_1-y_1| + |x_k-y_k| \big) \nonumber \\ & \quad  \le \frac{C\|\nabla\varphi\|_\infty}{r^{2d}\varphi_0 N} \bigg(|x_1-y_1| + |x_j-y_j| + \frac 1N \sum_{k=1}^N |x_k-y_k| \bigg) \quad \text{in the set }\mc G.
\end{align}
Therefore, from Cauchy-Schwarz inequality, in the set $\mc G$,
\[
\bigg(\sum_{j=1}^N |p_{1,j}-q_{1,j}||w_j|\bigg)^2 \le \frac{C\|\nabla\varphi\|_\infty^2}{r^{4d}\varphi_0^2} \bigg( |x_1-y_1|^2 + \frac 1N\sum_{k=1}^N |x_k-y_k|^2\bigg)\frac 1N \sum_{j=1}^N|w_j|^2 \,.
\]
Analogously, still in $\mc G$,
\[
\begin{split}
\big(\bar T_N^\varphi - \tilde T_N^\varphi\big)^2 \le & \bigg(\sum_{j=1}^N |p_{1,j} - q_{1,j}| |w_j-\tilde u_N^\varphi|^2\bigg)^2  \\ &\ \le \frac{C\|\nabla\varphi\|_\infty^2}{ r^{4d}\varphi_0^2} \bigg( |x_1-y_1|^2 + \frac 1N\sum_{k=1}^N |x_k-y_k|^2\bigg) \frac 1N \sum_{j=1}^N |w_j-\tilde u_N^\varphi|^4 \\ &\ \le \frac{C\|\nabla\varphi\|_\infty^2}{r^{5d}\varphi_0^2} \bigg( |x_1-y_1|^2 + \frac 1N\sum_{k=1}^N |x_k-y_k|^2\bigg) \frac 1N \sum_{j=1}^N |w_j|^4\,,
\end{split}
\]
where in the last inequality, we used that, because of Eq.~\eqref{stimdens},
\[
|\tilde u_N^\varphi\big|^4 \le \bigg(\sum_{j=1}^N q_{1,j} |w_j| \bigg)^4 \le \sum_{j=1}^Nq_{1,j} |w_j|^4 \le \frac{1}{Ar^dN} \sum_{j=1}^N |w_j|^4\quad \text{ in the set }\mc G.
\]
Recalling Eq.~\eqref{D2}, the above estimates allow to control  $D_2(Z_N,\Sigma_N)$ in the set $\mc G$,
\[
\mc D_c(t) \le  \frac{C\|\nabla\varphi\|_\infty^2}{r^{5d}\varphi_0^2}  \int\! \rmd R_N(t)\, \bigg( |x_1-y_1|^2 + \frac 1N\sum_{k=1}^N |x_k-y_k|^2\bigg) \frac 1N \sum_{j=1}^N (1+|w_j|^4)\,.
\]
We now argue analogously to what done to get Eq.~\eqref{db3}: recalling  Eq.~\eqref{eq:llnw} and inserting $1= \id_{\mc G_{M,4}} + \id_{\mc G_{M,4}^\complement}$ in the right-hand side we have that 
\begin{align}
\label{dc1}
\mc D_c(t) & \le \frac{2C(1+M)\|\nabla\varphi\|_\infty^2}{r^{5d}\varphi_0^2} I_N(t) + \frac{C\|\nabla\varphi\|_\infty^2}{r^{5d}\varphi_0^2}  \int\! \rmd R_N(t)\,\frac 1N \sum_{j=1}^N (1+|w_j|^4)  \id_{\mc G_{M,4}^\complement} \nonumber \\ & \le \frac{2C(1+M)\|\nabla\varphi\|_\infty^2}{r^{5d}\varphi_0^2} I_N(t) + \frac{C\|\nabla\varphi\|_\infty^2}{r^{5d} \varphi_0^2 N^{1/2}}\,,
\end{align}
where in estimating the integrand in $\mc G_{M,4}^\complement$ we used that the mutual distance among the particles is not greater than one.

\subsection{Upper bound on  \texorpdfstring{$\mc D_d(t)$}{d}}
\label{sec:5.4}

We decompose
\[
\mc D_d(t) = \mc D_d^{(1)}(t) + \mc D_d^{(2)}(t)\,,
\]
where 
\begin{align*}
D_d^{(1)}(t) & := \int\! \rmd g(t)^{\otimes N} \int\! \rmd \xi\, \varphi(\xi) \, 2 |\tilde u_N^\varphi- u_g^\varphi|^2\,, \\ \mc D_d^{(2)}(t) & := \int\! \rmd g(t)^{\otimes N} \int\! \rmd \xi\, \varphi(\xi) \,  \frac{4d\big( \tilde T_N^\varphi - T_g^\varphi \big)^2 }{A_t}\,,
\end{align*}
and analyze the two terms separately.

\medskip
\noindent
\textit{Upper bound on $\mc D_d^{(1)}(t)$.} After introducing the random variables
\[
\mc U_j(\xi) := \varrho_g^\varphi(\xi,t)\varphi(\xi-y_j)w_j - \varphi(\xi-y_j) (\varrho_g^\varphi u_g^\varphi) (\xi,t)\,,
\]
we observe that
\[
\begin{split}
& \int\! \rmd \xi\, \varphi(\xi) \,|\tilde u_N^\varphi- u_g^\varphi|^2 \le  \int\! \rmd \xi\, \varphi(\xi)\, \big[|\tilde u_N^\varphi- u_g^\varphi|^2 \id_{\mc B_A} + (|\tilde u_N^\varphi|^2 + |u_g^\varphi|^2) \id_{\mc B_A^\complement}\big] \\ & \quad \quad = \int\! \rmd \xi\, \frac{ \varphi(\xi-y_1) }{\tilde\varrho_N^\varphi (\xi)^2 \varrho_g^\varphi(\xi)^2} \bigg|\frac 1N \sum_{j=1}^N \mc U_j(\xi)  \bigg|^2\id_{\mc B_A} +  \int\! \rmd \xi\, \varphi(\xi) \big(|\tilde u_N^\varphi|^2 + |u_g^\varphi|^2\big) \id_{\mc B_A^\complement} \\ & \quad \quad \le  \frac{C}{\varphi_0^2 r^{2d}} \int\! \rmd \xi\, \varphi(\xi-y_1) \bigg|\frac 1N \sum_{j=1}^N \mc U_j(\xi)  \bigg|^2 +  C \int\! \rmd \xi\, \varphi(\xi)\sum_{j=1}^N q_{1,j} (1+|w_j|^2) \id_{\mc B_A^\complement}\,.
\end{split}
\]
Above, we applied (after the change of variables $\xi \to \xi+y_1$) the definition Eq.~\eqref{ba} and the lower bound Eq.~\eqref{stimrho} to estimate the first term in the right-hand side, the Cauchy-Schwarz inequality together with Eq.~\eqref{eq:utK}  to estimate  the second term.

We notice that the variables $\mc U_j(\xi)$ are i.i.d.\ and satisfy
\begin{equation}
\label{Uj}
\begin{split}
& \int\! \rmd g(t)^{\otimes N}\, \mc U_j(\xi) = 0\,, \\ & \int\! \rmd g(t)^{\otimes N}\, |\mc U_j(\xi)|^2 \le C \varphi_0 \int\!\rmd y\, \rmd w\, \varphi(\xi-y) g (y,w,t) |w|^2 \le  C \varphi_0\,,
\end{split}
\end{equation}
where we used the upper bound on $\mc N_q(g(t))$ given in Eq.~\eqref{eq:utK} with $q>2+d$. Therefore,
\begin{align}
\label{lln}
\int\! \rmd g(t)^{\otimes N}\,  \int\! \rmd \xi\, \varphi(\xi-y_1) \bigg|\frac 1N \sum_{j=1}^N \mc U_j(\xi)  \bigg|^2 & \le \varphi_0 \int\! \rmd \xi\, \int\! \rmd g(t)^{\otimes N} \bigg|\frac 1N \sum_{j=1}^N \mc U_j(\xi)  \bigg|^2 \nonumber \\ &  \le \frac{C\varphi_0^2}N\,.
\end{align}
On the other hand, 
\begin{align}
\label{stw}
& \int\! \rmd g(t)^{\otimes N}\,  \int\! \rmd \xi\, \varphi(\xi)\sum_{j=1}^N q_{1,j} (1+|w_j|^2) \id_{\mc B_A^\complement} \nonumber \\ & \qquad\qquad= \frac 1N \sum_{i=1}^N \int\! \rmd g(t)^{\otimes N} \int\!\rmd \xi\, \varphi(\xi) \sum_{j=1}^Nq_{i,j} (1 + |w_j|^2) \id_{\mc B_A^\complement}\nonumber \\ & \qquad \qquad \le C\varphi_0  \int\! \rmd g(t)^{\otimes N} \frac 1N \sum_{j=1}^N (1 + |w_j|^2) \id_{\mc B_A^\complement} \le \frac{C\varphi_0}{(r^{3d} N)^{1/2}}\,,
\end{align}
where we used that $\rmd g(t)^{\otimes N}$ and $\id_{\mc B_A^\complement}$ are symmetric with respect to particle permutations, the upper bound of Eqs.~\eqref{qijstim}, and finally, as done in Eq.~\eqref{db2}, the Cauchy-Schwarz inequality and  Eqs.~\eqref{eq:densy} and \eqref{eq:stimv}. 

Putting the above together, we obtain
\begin{equation}
\label{dd1}
\mc D_d^{(1)}(t) \le C  \bigg(\frac{1}{r^{2d}N} + \frac{\varphi_0}{(r^{3d} N)^{1/2}}\bigg)\,.
\end{equation}

\medskip
\noindent
\textit{Upper bound on $\mc D_d^{(2)}(t)$.} We argue similarly to the previous case. Since
\[
d(\tilde T_N^\varphi - T_g^\varphi) = \sum_{j=1}^N q_{1,j} |w_j|^2 - (d T_g^\varphi + |u_g^\varphi|^2) + |u_g^\varphi|^2 - |\tilde u_N^\varphi|^2\,,
\]
after introducing the random variables
\[
\mc T_j(\xi) = \varrho_g^\varphi(\xi,t)\varphi(\xi-y_j) |w_j|^2 - \varphi(\xi-y_j) ( \varrho_g^\varphi d T_g^\varphi + |u_g^\varphi|^2)(\xi,t)
\]
and using that $(|u_g^\varphi|^2 - |\tilde u_N^\varphi|^2)^2 \le |\tilde u_N^\varphi- u_g^\varphi|^4$, we have
\begin{align}
\label{t-t}
& \int\! \rmd \xi\, \varphi(\xi) \, (\tilde T_N^\varphi - T_g^\varphi)^2  \le \int\! \rmd \xi\, \varphi(\xi)\, \big[(\tilde T_N^\varphi - T_g^\varphi)^2 \id_{\mc B_A} + \big((\tilde T_N^\varphi)^2 + (T_g^\varphi)^2\big)\id_{\mc B_A^\complement}\big]  \nonumber \\ & \quad  = \int\! \rmd \xi\, \frac{\varphi(\xi-y_1)}{\tilde\varrho_N^\varphi (\xi)^2 \varrho_g^\varphi(\xi)^2} \bigg\{\bigg(\frac 1{Nd} \sum_{j=1}^N \mc T_j(\xi)  \bigg)^2 + \frac{1}{\tilde\varrho_N^\varphi (\xi)^2 \varrho_g^\varphi(\xi)^2} \bigg|\frac 1{Nd} \sum_{j=1}^N \mc U_j(\xi)  \bigg|^4 \bigg\} \id_{\mc B_A} \nonumber \\ & \quad \qquad + \int\! \rmd \xi\, \varphi(\xi) \big[(\tilde T_N^\varphi)^2 + (T_g^\varphi)^2\big]\id_{\mc B_A^\complement} \nonumber \\ & \quad \quad \le  \frac{C}{\varphi_0^2 r^{2d}} \int\! \rmd \xi\, \varphi(\xi-y_1)\bigg\{\bigg(\frac 1N \sum_{j=1}^N \mc T_j(\xi)  \bigg)^2 +  \frac{1}{\varphi_0^2 r^{2d}} \bigg|\frac 1N \sum_{j=1}^N \mc U_j(\xi)  \bigg|^4\bigg\}  \nonumber \\ & \quad\qquad + C \int\! \rmd \xi\, \varphi(\xi)\sum_{j=1}^N q_{1,j} (1+|w_j|^4) \id_{\mc B_A^\complement}\,.
\end{align}
Clearly, the expectation with respect to $\rmd g(t)^{\otimes N}$ of the last term in the right-hand side of Eq.~\eqref{t-t} can be bounded as in Eq.~\eqref{stw}, while, by Eq.~\eqref{Uj},
\[
\int\! \rmd g(t)^{\otimes N}\,  \int\! \rmd \xi\, \varphi(\xi-y_1) \bigg|\frac 1N \sum_{j=1}^N \mc U_j(\xi)  \bigg|^4 \le \frac{C\varphi_0^3}{N^2}\,.
\]
Finally, concerning the expectation with respect to $\rmd g(t)^{\otimes N}$ of the first term in the right-hand side of Eq.~\eqref{t-t}, we observe that also the variables $\mc T_j(\xi)$ are i.i.d.\ and satisfy
\[
\begin{split}
& \int\! \rmd g(t)^{\otimes N}\, \mc T_j(\xi) = 0\,, \\ & \int\! \rmd g(t)^{\otimes N}\, |\mc T_j(\xi)|^2 \le C \varphi_0 \int\!\rmd y\, \rmd w\, \varphi(\xi-y) g (y,w) |w|^4 \le  C \varphi_0\,,
\end{split}
\]
where we used the upper bound on $\mc N_q(g(t))$ given in Eq.~\eqref{eq:utK} with $q>4+d$. Therefore, analogously to Eq.~\eqref{lln},
\[
\int\! \rmd g(t)^{\otimes N}\,  \int\! \rmd \xi\, \varphi(\xi-y_1) \bigg(\frac 1N \sum_{j=1}^N \mc T_j(\xi)  \bigg)^2 \le \frac{C\varphi_0^2}N\,.
\]

Collecting together the above bounds, we obtain
\begin{equation}
\label{dd2}
\mc D_d^{(2)}(t) \le C  \bigg(\frac{1}{r^{2d}N} +  \frac{1}{r^{4d}\varphi_0N^2} + \frac{\varphi_0}{(r^{3d} N)^{1/2}}\bigg)\,.
\end{equation}

From Eqs.~\eqref{dd1} and \eqref{dd2} we conclude that
\begin{equation}
\label{dd5}
\mc D_d(t) \le C \bigg(\frac{1}{r^{2d}N} +  \frac{1}{r^{4d}\varphi_0 N^2} + \frac{\varphi_0}{(r^{3d} N)^{1/2}}\bigg)\,.
\end{equation}

\subsection{Proof of Eq.~(\ref{I_N})}
\label{sec:5.5}

From Eqs.~\eqref{I<2}, \eqref{da}, \eqref{da<}, \eqref{dc1}, and \eqref{dd5} we get
\[
\begin{split}
\dot I_N(t) & \le C \bigg( \varphi_0 + \frac{\|\nabla\varphi\|_\infty}{r^{2d}\varphi_0} + \frac{\|\nabla\varphi\|_\infty^2}{r^{5d}\varphi_0^2} \bigg) I_N(t) \\ & \quad + \bigg(\frac{1}{r^{2d} N} +  \frac{1}{r^{4d}\varphi_0 N^2} +  \frac{\varphi_0}{(r^{3d} N)^{1/2}} + \frac{\|\nabla\varphi\|_\infty^2}{r^{5d} \varphi_0^2 N^{1/2}} \bigg)\,.
\end{split}
\]
We then choose
\begin{equation}
\label{N0Gammaphi}
\Gamma_{\varphi} = \frac{\varphi_0^3 + \|\nabla\varphi\|_\infty^2}{r^{5d}\varphi_0^2}\,, \qquad N_\varphi = \frac{\varphi_0^4} {r^{6d}} + \frac{\|\nabla\varphi\|_\infty^8}{(r^{5d}\varphi_0^2)^4}\,,
\end{equation}
so that, recalling $\varphi_0\ge1$,
\[
\bigg(\varphi_0 + \frac{\|\nabla\varphi\|_\infty}{r^{2d}\varphi_0} +  \frac{\|\nabla\varphi\|_\infty^2}{r^{5d}\varphi_0^2} \bigg) \le 2\Gamma_\varphi
\]
and
\[
\bigg(\frac{1}{r^{2d}N} +  \frac{1}{r^{4d}\varphi_0N^2} + \frac{\varphi_0}{(r^{3d} N)^{1/2}} + \frac{\|\nabla\varphi\|_\infty^2}{r^{5d} \varphi_0^2 N^{1/2}} \bigg) \le \frac{4}{N^{1/4}} \qquad \forall\, N>N_\varphi\,.
\]
Therefore, $\dot I_N(t) \le C\Gamma_\varphi I_N(t) + C/N^{1/4}$ for any $N>N_\varphi$, from which Eq.~\eqref{I_N} follows by Gr\"{o}nwall's inequality and hence the proof of Theorem \ref {teo:main} is achieved.

\section{Removing the cutoff and conclusion}
\label{sec:6}

The regularized BGK equation Eq.~\eqref{eq:kin} reduces to the usual one Eq.~\eqref{BGK} when the cutoff function $\varphi$ converges to the $\delta$-function, at least formally. In this section, $\varphi=\varphi^{(\eps)}$ is the rescaled function given by \eqref{scaphi} and we assume $t\in [0,T]$ for fixed $T>0$.  We recall that $C$ denotes a generic positive constant, possibly depending only on $T$ and $f_0$, hence independent of $\eps$.

\begin{proof}[Proof of  Theorem \ref {teo:lim}]
The limit $\eps \to 0$ was investigated in \cite{BHP}, where the convergence $g \to f$ is proven in a weighted $L^1$ space, see \cite[Theorem 2.4]{BHP}. Here, given $ f$ and $g$, we rather study the processes $(x(t),v(t))$ and $(y(t),w(t))$ whose generators are given by ($\psi$ denotes a test function)
\begin{equation}
\label{nonlin1}
\mc L_1\psi(x,v) = (v\cdot \nabla_x -1)\psi (x,v) + \int\! \rmd \tilde v\, M_f( x,\tilde v) \psi(x,\tilde v)
\end{equation}
and $\mc L_2 = \mc L_1^g$ as in Eq.~\eqref{nonlin}, i.e.,
\begin{equation}
\label{nonlin2}
\mc L_2\psi(y,w) = (w\cdot \nabla_y -1) \psi(y,w) + \int\! \rmd \tilde y \, \varphi(\tilde y - y) \int\! \rmd \tilde w\, M_g^\varphi(\tilde y,\tilde w) \psi(\tilde y,\tilde w) \,,
\end{equation}
respectively.

We fix $f_0$ as common initial datum for Eqs.~\eqref{eq:bgk} and \eqref{eq:kin}, and assume that it satisfies the hypotheses Eqs.~\eqref{eq:f0} and \eqref{grad}. Note that, in particular, the hypothesis of Propositions \ref{prop:bgk} and \ref{prop:stim_uT} are satisfied.

We couple the two processes by defining in the product space the process whose generator is given by ($G$ denotes a test function)
\[
\begin{split}
\mc L_Q G (x,v,y,w)  & = ( v\cdot \nabla_x+w\cdot \nabla_y -1) G(x,v,y,w) \\ & \quad + \int\! \rmd \tilde y \int\! \rmd \tilde v \int\! \rmd \tilde w\, \mc M
(x,\tilde v;\tilde y, \tilde w)  \varphi(\tilde y - y)  G (x, \tilde v , \tilde y,\tilde w)\,,
\end{split} 
\]
where, analogously to what done in Sec.~\ref{sec:3.1}, $\mc M(x,\tilde v;\tilde y, \tilde w)$ is the joint representation of the Maxwellians $M_f$ and $M^\varphi_g$ that realizes the 2-Wasserstein distance between the marginals. We denote by $\rmd R(t) = \rmd R(x,v;y,w;t)$ the law of this process assuming that initially
\[
\rmd R(x,v;y,w;0) = f_0(x,v) \delta (x-y) \delta (v-w).
\] 

Clearly,
\begin{equation}
\label{W<I}
\mc W_2(f(t),g(t))^2 \le I(t) := \int\! \rmd R(t)\, \big(|x-y|^2 + |v-w|^2\big).
\end{equation}
To obtain an upper bound for $I(t)$ we compute,
\[
\begin{split}
\dot I(t) & = 2\int\! \rmd R(t)\, (x-y) \cdot (v-w) -\int\! \rmd R(t)\, |x-y|^2 \\ & \quad + \int\! \rmd R(t) \int\! \rmd \tilde y \, \varphi(y-\tilde y) |x-\tilde y|^2 + \frac{\rmd}{\rmd t} \int\! \rmd R(t)\, |v-w|^2 \\ & \le \int\! \rmd R(t)\, |v-w|^2 + 2 \int\! \rmd R(t)\, |x-y|^2 + E_1 + \frac{\rmd}{\rmd t} \int\! \rmd R(t)\, |v-w|^2 \,,
\end{split}
\]
where
\[
E_1= 2\int\! \rmd R(t) \int\! \rmd \tilde y\,  \varphi(y -\tilde y ) |y-\tilde y|^2  \le C \eps^2\,,
\]
while, by virtue of the choice of $\mc M(x,\tilde v;\tilde y, \tilde w)$,
\[
\frac{\rmd}{\rmd t} \int\! \rmd R(t)\, |v-w|^2 = S_1 +  S_2 -  \int\! \rmd R(t)\, |v-w|^2\,,
\]
with 
\[
S_1 = \int\! \rmd R(t)\, |u_f(x,t) -u_g^\varphi (y,t) |^2\,, \quad S_2 = \int\! \rmd R(t)\, d \Big(\sqrt {T_f(x,t)} - \sqrt {T_g^\varphi (y,t)}\Big)^2\,.
\]
Therefore,
\begin{equation}
\label{I<}
\dot I(t) \le 2 I(t) + S_1 + S_2 + C \eps^2
\end{equation}
and it remains to estimate $S_1$ and $S_2$.

\medskip
\noindent
\textit{Upper bound on $S_1$.} Setting 
\[
\begin{split}
\rmd r(x,y;t) & = \int_{v,w}\! \rmd R(x,v;y,w;t)\,, \\ \rmd R^\varphi(x,v;y,w;t)  & = \int\! \rmd \tilde y\, \varphi(y-\tilde y) \, \rmd R(x,v;\tilde y,w;t) \,,\\ \rmd r^\varphi(x,y;t) & = \int\! \rmd \tilde y\, \varphi(y-\tilde y)\, \rmd r(x,\tilde y;t)\,,
\end{split}
\]
we have
\[
S_1=\bar S_1+E_2\,,
\]
where
\[
\bar S_1= \int\! \rmd r^\varphi(t)\,  |u_f(x,t) -u_g^\varphi (y,t) |^2\,, \quad
E_2= \int\! (\rmd r- \rmd r^\varphi) (t)\, |u_f(x,t) -u_g^\varphi (y,t) |^2\,.
\]
We start by estimating $\bar S_1$, noticing that if
\[
g^\varphi(y,w,t) := \int\!\rmd \tilde y\, \varphi(y-\tilde y) g(\tilde y, w,t)
\]
then
\begin{align}
\label{T1}
\bar S_1= & \int\! \rmd r^\varphi(t)\,  \bigg|\frac {\int\!\rmd v\, f(x,v,t) v}{\varrho_f(x,t)} - \frac {\int\!\rmd w\, g^\varphi(y,w,t)w}{ \varrho_g^\varphi(y,t)}\bigg|^2 \nonumber  \\ = &
\int\! \rmd r^\varphi(t)\,\bigg|\int\! \rmd \Lambda_{x,y}(v,w)\, (v-w) \bigg|^2 \le
\int\! \rmd r^\varphi(t) \int\! \rmd \Lambda_{x,y}(v,w)\, |v-w|^2\,,
\end{align}
with (remarking also that $\varrho_g^\varphi = \varrho_{g^\varphi}$)
\begin{equation}
\label{rc}
\rmd \Lambda_{x,y}(v,w)\in \mc P\bigg(\frac {f(x,v,t)\, \rmd v}{\varrho_f(x,t)},\frac{g^\varphi(y,w,t) \,\rmd w}{\varrho_g^\varphi(y,t)}\bigg)
\end{equation}
to be fixed later (recall $\mc P(\mu,\nu)$ denotes the collection of all joint probability measures in the product space with marginals $\mu$ and $\nu$, see Remark \ref{rem:wass}).
 
Looking at the coupled stochastic processes $(x(t),y(t))$, we realize it is of the form 
\[
\begin{split}
x(t) & = x_0 + v_0 t_1 + v_1 (t_2-t_1) +v_2 (t_3-t_2) \cdots\,,
\\ y(t) & = x_0 + v_0 t_1 + \xi_1+ w_1 (t_2-t_1) + \xi_2 + w_2 (t_3-t_2) \cdots
\end{split}
\]
where $t_1 <t_2 < \cdots < t_k $ are exponential times in which the jumps in velocity are simultaneously performed. The outgoing velocities are Maxwellian computed via the hydrodynamical fields depending on $x(t_k)$ and $y(t_k) +\xi_k$. Finally, the extra displacements $\xi_k$ are i.i.d.~distributed according to $\varphi$. 

Now, let $R(\rmd v\, \rmd w|x,y;t)$, resp.~$R^\varphi(\rmd v\, \rmd w|x,y;t)$, be the conditional probability of $\rmd R(x,v;y,w;t)$, resp.~$\rmd R^\varphi (x,v;y,w;t)$, conditioned to the values $x,y$ at time $t$. For what just noticed on the structure of the coupled process, the conditional probability has the property that $\int\! R(\rmd v\, \rmd w|x,y;t) \, \psi (v)$ and $\int\! R(\rmd v\, \rmd w|x,y;t)\, \psi (w)$ are independent of $y$ and $x$, respectively. Similarly, also $\int\! R^\varphi(\rmd v\, \rmd w|x,y;t)\, \psi (v)$ and $\int\! R^\varphi(\rmd v\, \rmd w|x,y;t)\, \psi (w)$ are independent of $y$ and $x$, respectively.

\begin{remark}
\label{rem:ma}
Note that the measures $\rmd R (x,v;y,w;t)$ and $\rmd r (x,y;t) $ are not absolutely continuous with respect to the Lebesgue measure. Indeed, the contribution due to the event with zero jumps transports the initial delta function in position and velocity. Incidentally, this event does not give any contribution in the evaluation of $I(t)$. On the other hand, due to the convolution with $\varphi$,  $\rmd r^\varphi$ is absolutely continuous with respect to the Lebesgue measure, and we denote by $r^\varphi$ its density.
\end{remark}

We now observe that
\[
R^\varphi(\rmd v\, \rmd w|x,y;t) \in \mc P\bigg(\frac {f(x,v,t)\, \rmd v}{\varrho_f(x,t)},\frac{g^\varphi(y,w,t) \,\rmd w}{\varrho_g^\varphi(y,t)}\bigg)\,.
\]
Indeed, since $\varrho_g^\varphi(y,t)^{-1} \int\! \rmd x\, r^\varphi(x,y;t) =1$, for any test function $\psi$,
\[
\begin{split}
& \int\! \rmd y \int\! R^\varphi(\rmd v\, \rmd w|x,y;t)\, \psi(y,w) =  \int\! \rmd x \int\! \rmd y\, \frac{r^\varphi(x,y;t)}{\varrho_g^\varphi(y,t)} \int\! R^\varphi(\rmd v\, \rmd w|x,y;t)\, \psi(y,w) \\
& \qquad \quad = \int\! \rmd R^\varphi(x,v;y,w;t)\, \frac{\psi(y,w)}{\varrho_g^\varphi(y,t)}= \int\!\rmd y \int\! \frac{g^\varphi(y,w,t) \,\rmd w}{\varrho_g^\varphi(y,t)} \psi(y,w)\,,
\end{split}
\]
where, in the first step, we used the independence of $x$ of the left-hand side. The same argument can be used to prove that
\[
\int\!\rmd x \int\!R^\varphi(\rmd v\, \rmd w|x,y;t)\, \psi(x,v)  = \int\!\rmd x \int\!\frac {f(x,v,t)\, \rmd v}{\varrho_f(x,t)} \psi(x,v)\,.
\]
By choosing $\rmd \Lambda_{x,y}(v,w) = R^\varphi(\rmd v \, \rmd w|x,y;t)$ in Eq.~\eqref{T1} we obtain
\begin{align}
\label{est:T1}
\bar S_1 & \le  \int\! \rmd r^\varphi(t) \int\! R^\varphi(\rmd v \, \rmd w|x,y;t) \, |v-w|^2 \nonumber \\ & = \int\!\rmd R^\varphi(t)\, |v-w|^2 =  \int\!\rmd R (t)\, |v-w|^2\,.
\end{align}

Concerning the error term $E_2$, since $\int (\rmd r- \rmd r^\varphi) (t)\, |u_f(x)|^2 =0$ we have 
\[
E_2 = \int\! (\rmd r- \rmd r^\varphi) (t)\, \big(|u_g^\varphi(y,t)|^2 - 2 u_f(x,t) \cdot u_g^\varphi (y,t)\big) = E_2^{(1)} +E_2^{(2)}\,,
\]
with 
\[
\begin{split}
E_2^{(1)} & = \int\! \rmd r(t) \int\!\rmd \tilde y\, \varphi(y-\tilde y) \big(|u_g^\varphi(y,t)|^2 - |u_g^\varphi(\tilde y,t)|^2\big)\,, \\ E_2^{(2)} & = - 2 \int\! \rmd r(t) \int\!\rmd \tilde y\, \varphi(y-\tilde y) u_f(x,t) \cdot \big(u_g^\varphi(y,t) - u_g^\varphi(\tilde y,t) \big)\,.
\end{split}
\]
The assumptions Eqs.~\eqref{eq:f0} and \eqref{grad} are the hypotheses of \cite[Lemma 4.1]{BHP}, which in particular states
\[
\mc N_q(|\nabla_x g(t)|) \le C\,.
\]
This implies (see the proof of the same lemma)
\begin{equation}
\label{gradstim}
|\nabla_x \varrho_g| + |D_x u_g| + |\nabla_x T_g| \le C\,.
\end{equation}
Therefore, using also Eqs.~\eqref{eq:utf} and \eqref{eq:utK},
\[
\big|E_2^{(1)} \big| + \big|E_2^{(2)} \big| \le C  \int\! \rmd r(t)  \int\! \rmd \tilde y\,  \varphi(y -\tilde y ) |y-\tilde y| \le C \eps\,.
\]
In conclusion, from Eq.~\eqref{est:T1} and the above estimate,
\begin{equation}
\label{S1}
S_1 \le  C \int\! \rmd R(t)\, |v-w|^2 + C\eps \le I(t) + C\eps \,.
\end{equation}

\medskip
\noindent
\textit{Upper bound on $S_2$.} We proceed  analogously by setting
\[
S_2=\bar S_2+E_3\,,
\]
where
\[
\bar S_2 = \int\! \rmd r^\varphi(t)\, d \Big(\sqrt {T_f(x,t)} - \sqrt {T_g^\varphi (y,t)}\Big)^2
\]
and, by arguing as done for $E_2$,
\[
\begin{split}
E_3 & = \int\! (\rmd r- \rmd r^\varphi) (t)\, \Big(T_g^\varphi(y,t) - 2 \sqrt{T_f(x,t)}\sqrt{T_g^\varphi(y,t)}\Big) \\ & = \int\! \rmd r(t) \int\!\rmd \tilde y\, \varphi(y-\tilde y) \big(T_g^\varphi(y,t) - T_g^\varphi(\tilde y,t)\big) \\ & \quad - 2  \int\! \rmd r(t) \int\!\rmd \tilde y\, \varphi(y-\tilde y) \sqrt{T_f(x,t)}\Big(\sqrt{T_g^\varphi(y,t)} - \sqrt{T_g^\varphi(\tilde y,t)}\Big)\,.
\end{split}
\]
From Eqs.~\eqref{gradstim} and \eqref{eq:TA} the error term $E_3$ satisfies
\begin{equation}
\label{E3}
|E_3| \le C \int\! \rmd r(t)  \int\! \rmd \tilde y\,  \varphi(y -\tilde y ) |y-\tilde y| \le C \eps\,.
\end{equation}

Concerning $\bar S_2$, we observe that (omitting for brevity the dependence on $(x,t)$ in $u_f(x,t), T_f(x,t)$ and on $(y,t)$ in $u_g^\varphi(y,t), T_g^\varphi(y,t)$)
\[
\begin{split}
|T_f -T^\varphi_g| & = \frac 1d \bigg| \int\! \rmd \Lambda_{x,y} (v,w)\, \Big( |v-u_f|^2 - |w-u^\varphi_g|^2 \Big) \bigg| \\ & = \frac 1d  \bigg| \int\! \rmd \Lambda_{x,y} (v,w)\, \big(v-w + u_g^\varphi - u_f\big) \cdot  \big(v- u_f+  w - u_g^\varphi \big) \bigg| \\ & \le  \frac 2d \bigg[\int\! \rmd \Lambda_{x,y} (v,w)\, \big(|v-w|^2 + |u_g^\varphi - u_f|^2\big)\bigg]^{1/2} \\& \qquad \times \bigg[\int\! \rmd \Lambda_{x,y} (v,w)\, \big(|v - u_f|^2 + |w - u_g^\varphi|^2\big)\bigg]^{1/2} \\ & \le C \bigg[\int\! \rmd \Lambda_{x,y} (v,w)\, \big(|v-w|^2 + |u_g^\varphi - u_f|^2\big)\bigg]^{1/2}
\big( \sqrt {T_f} + \sqrt {T^\varphi_g} \big)\,.
\end{split}
\]
where in the first and last step we used \eqref{rc}. Therefore,
\[
\begin{split}
\bar S_2 & = \int\! \rmd r^\varphi(t)\, d \bigg(\frac{T_f(x,t) -T_g^\varphi (y,t)}{\sqrt {T_f(x,t)} + \sqrt {T_g^\varphi (y,t)}}\bigg)^2 \nonumber  \\ & \le C  \int\! \rmd r^\varphi(t) \int\! \rmd \Lambda_{x,y} (v,w)\, \big(|v-w|^2 + |u_g^\varphi(y,t) - u_f(x,t)|^2\big) \nonumber \\ & = C \int\! \rmd R(t)\, |v-w|^2 + C \int\! \rmd R(t)\,  |u_g^\varphi(y,t) - u_f(x,t)|^2  \nonumber \\ & = C \int\! \rmd R(t)\, |v-w|^2 + C S_1\,.
\end{split}
\]
Therefore, from Eq.~\eqref{E3} and the above estimate,
\begin{equation}
\label{S2}
S_2 \le C I(t) + C S_1 + C\eps\,.
\end{equation}

\smallskip
From Eqs.~\eqref{W<I}, \eqref{I<}, \eqref{S1}, \eqref{S2}, and Gr\"{o}nwall's inequality, Eq.~\eqref{conveps} follows and Theorem \ref{teo:lim} is proved.
\end{proof}

\begin{proof}[Proof of  Theorem \ref {teo:main1}]
We have,
\[
\begin{split}
\mc W_2 ( f^N_j(t), f(t) ^{\otimes j} ) & \leq \mc W_2 ( f^N_j(t), g(t)^{\otimes j} )+ \mc W_2(g(t)^{\otimes j}, f(t) ^{\otimes j}) \\ & \le C \sqrt j \bigg(\frac 1{N^{1/8}} \exp \frac{C(\varphi_0^3 + \|\nabla\varphi\|_\infty^2)}{r^{5d}\varphi_0^2}+ \sqrt\eps \bigg),
\end{split}
\]
after using Theorem \ref {teo:lim}, Theorem \ref{teo:main}, and inserting the expression of $\Gamma_\varphi$ given in Eq.~\eqref{N0Gammaphi}.

Since $ \|\nabla \varphi \|_\infty \le  C\eps^{-(d+1)}$,  $ \varphi_0 \le C\eps^{-d}$,  $r>C\eps$ (in order to satisfy Eq.~\eqref{r0}), and $N>N_\varphi$, the limit Eq.~\eqref{conv} follows provided $\eps$ is vanishing slowly when $N$ diverges (e.g., $\eps=(\log N)^{-\mu}$ with $\mu$ sufficiently small). The theorem is thus proved.  
\end{proof}


\end{document}